\documentclass[11pt]{article}
\usepackage{graphicx} 
\usepackage{color}
\usepackage{amsfonts}
\usepackage{tkz-euclide}
\usepackage{amsthm}
\usepackage{amsmath}
\usepackage{amssymb}
\usepackage[toc]{appendix}
\usepackage[a4paper, margin=1in]{geometry}
\usepackage{multirow}
\usepackage{adjustbox}
\usepackage{algorithm}
\usepackage{algpseudocode}
\usepackage{hyperref}
\usepackage{orcidlink}
\usepackage{subcaption}

\title{Seat Arrangement Problems under B-utility and W-utility}
\author{José Rodríguez \orcidlink{0009-0008-2470-6549}}

\begin{document}

\newtheorem{theorem}{Theorem}
\newtheorem{proposition}{Proposition}
\newtheorem{lemma}{Lemma}
\newtheorem{claim}{Claim}
\newtheorem{observation}{Observation}
\newtheorem{corollary}{Corollary}

\maketitle

\begin{abstract}
In the \textsc{seat arrangement} problem the goal is to allocate agents to vertices in a graph such that the resulting arrangement is optimal or fair in some way. Examples include an arrangement that maximises utility or one where no agent envies another. We introduce two new ways of calculating the utility that each agent derives from a given arrangement, one in which agents care only about their most preferred neighbour under a given arrangement, and another in which they only care about their least preferred neighbour. We also present a restriction on agent's preferences, namely $1$-dimensional preferences. We give algorithms, hardness results, and impossibility results for these types of utilities and agents' preferences. Additionally, we refine previous complexity results, by showing that they hold in more restricted settings.
\end{abstract}

\section{Introduction}
\label{sec:Intro}
Let us imagine that we are organizing a wedding. An important part of the event is the dinner, and for that we need to decide where our guests are seated. However, guests have preferences over who sits next to them. For example, couples want to be together, parents may prefer to be next to their kids, and kids may want to be next to other kids. We may even have the case of two relatives who dislike each other, and thus would rather not be seated together. The problem of deciding how to best allocate these wedding guests to seats is known as \textsc{seat arrangement} \cite{bodlaender_hedonic_2020}.

An instance of \textsc{seat arrangement} consists of a set of agents, who have cardinal valuations over each other, and a graph with exactly as many vertices as there are agents, where each vertex represents a ``seat" that some agent will be allocated to. Our goal is to find an arrangement, i.e., an allocation from agents to seats, that is optimal or fair in a precise sense according to the preferences of agents over each other.

In a given arrangement, each agent has a utility, which depends on who is next to the agent, where ``next'' corresponds to the neighbourhood in the seat graph. Given this, we can define four types of desirable arrangements and the associated problems of computing these arrangements or reporting that they do not exist: we can find an arrangement where the sum of utilities of all agents is highest, whose associated computational problem is \textsc{Maximum Welfare Arrangement (MWA)}; an arrangement where the utility of the worst-off agent is maximal, whose associated computational problem is \textsc{Maximin Utility Arrangement (MUA)}; an arrangement that is \textit{envy-free}, meaning that no agent wants to swap places with another agent, whose associated computational problem is \textsc{Envy-Free Arrangement (EFA)}; and an arrangement that is \textit{exchange-stable}, meaning that no two agents want to swap places with each other, whose associated computational problem is \textsc{Stable Arrangement (STA)}.

Up until now, the only type of utility that had been considered is one where the utility of an agent equals the sum of the valuations that said agent gives to each of their neighbours, which we call \textit{S-utility}. However, depending on the setting, agents may only care about who their best neighbour is, or conversely who their worst neighbour is. For example, in a wedding, guests may be happy being seated next to only one person whom they like, as each guest could spend the dinner talking to their favourite neighbour and ignore the others. On the other hand, in an office setting, employers may be keen to avoid sitting two employees who dislike each other together, even if each of them is surrounded by coworkers whom they like, as these two problematic employees may be unhappy being next to each other, and their work could suffer. In this work we consider these two types of utilities, referred to as \textit{B-utility} and \textit{W-utility} respectively, that are new concepts in the context of \textsc{seat arrangement}.

Additionally, because most problems related to \textsc{seat arrangement} are computationally hard, we present a restriction on the preferences of agents, referred to as $1$-dimensional preferences. Under this restriction, agents' preferences must be derived from a $1$-dimensional space. An example would be political ideology, where left-wing agents prefer other left-wing agents over right-wing agents, and similarly for agents located at other points in the left-right axis. Other restrictions on preferences that we study are strict preferences, where no agent can give the same valuation to two different agents; binary preferences, in which the only allowed valuations are zero and one; and symmetric preferences, where the valuation that an agent $p$ gives to an agent $q$ is equal to the valuation that $q$ gives to $p$.

We give both algorithms and hardness results for all four types of optimal or fair arrangements that we are interested in.

\subsection{Our contribution}
In this paper we introduce two different ways of calculating the utility that an agent derives from a seating arrangement: B-utility and W-utility. We study the problems of finding optimal arrangements in instances that have one of these two utilities, and present both algorithms and computational hardness results. In particular, we show that with B-utility and W-utility, MWA, MUA, and EFA are NP-hard, even if preferences are symmetric and binary, or symmetric and strict. On the other hand, if preferences are symmetric, we prove that with both B-utility or W-utility an exchange-stable arrangement is guaranteed to exist, although only in the first case the algorithm is shown to run in polynomial time.

\textsc{seat arrangement} with the most common utility, S-utility, is also studied, and many previous complexity results are refined. Thus, we prove that, with S-utility, MWA, MUA and EFA are NP-hard even if preferences are symmetric and strict. Another contribution is the introduction of a type of agents' preferences, $1$-dimensional preferences. We show that if the seat graph has a connected component which is a cycle with at least four vertices, or a path with at least three vertices, then no envy-free arrangement exists, if the utility type of the instance is S-utility. If the utility type is B-utility, and preferences are strict and $1$-dimensional, MWA, MUA, and EFA are NP-hard.

Table \ref{tab:resultSummary} provides a summary of the known results of \textsc{seat arrangement}, with our contributions included. If a given result is included in this paper, its corresponding cells will have, within square brackets, a number preceded by a letter, which indicates if the result is a theorem (T), an observation (O), or a proposition (P).

\begin{table}[h]
\begin{adjustbox}{width=1\textwidth}
\begin{tabular}{ll|lll|lll|lll|lll}
\hline
\multicolumn{2}{l}{}              & \multicolumn{3}{|c}{MWA}                                                              & \multicolumn{3}{|c}{MUA}                                                              & \multicolumn{3}{|c}{EFA}                                                                            & \multicolumn{3}{|c}{STA}                                                                                   \\ \hline
Val.                  & Graph    & B                          & S                          & W                          & B                          & S                          & W                          & B                          & S                                 & W                                 & B                                   & S                                & W                                \\ \hline
                       & Arb.     & {\color[HTML]{FE0000} N} [T \ref{thm:MaxMaxminSymm}] & {\color[HTML]{FE0000} N} \cite{bodlaender_hedonic_2020} & {\color[HTML]{FE0000} N} [T \ref{thm:MaxMaxminSymm}] & {\color[HTML]{FE0000} N} [T \ref{thm:MaxMaxminSymm}] & {\color[HTML]{FE0000} N} \cite{bodlaender_hedonic_2020} & {\color[HTML]{FE0000} N} [T \ref{thm:MaxMaxminSymm}] & {\color[HTML]{FE0000} N} [T \ref{thm:EfbUtilityTies}] & {\color[HTML]{FE0000} N}\cite{massand_graphical_2019}        & {\color[HTML]{FE0000} N}  [T \ref{thm:EnvyWutilSymmBinary}]      & {\color[HTML]{FE0000} N} [O \ref{obs:MatchPol}]         & {\color[HTML]{FE0000} N} \cite{bodlaender_hedonic_2020, cechlarova_complexity_2002, cechlarova_exchange-stable_2005}      & {\color[HTML]{FE0000} N} [O \ref{obs:MatchPol}]      \\ \cline{2-14} 
\multirow{-2}{*}{Arb.} & Match.   & {\color[HTML]{3166FF} P}[O \ref{obs:MatchPol}]   & {\color[HTML]{3166FF} P} \cite{bodlaender_hedonic_2020}  & {\color[HTML]{3166FF} P}[O \ref{obs:MatchPol}]   & {\color[HTML]{3166FF} P} [O \ref{obs:MatchPol}]  & {\color[HTML]{3166FF} P} \cite{bodlaender_hedonic_2020}  & {\color[HTML]{3166FF} P} [O \ref{obs:MatchPol}]  & {\color[HTML]{3166FF} P} [O \ref{obs:MatchPol}]  & {\color[HTML]{3166FF} P}  \cite{bodlaender_hedonic_2020}        & {\color[HTML]{3166FF} P} [O \ref{obs:MatchPol}]         & {\color[HTML]{FE0000} N}  [O \ref{obs:MatchPol}]      & {\color[HTML]{FE0000} N}  \cite{bodlaender_hedonic_2020, cechlarova_complexity_2002, cechlarova_exchange-stable_2005}     & {\color[HTML]{FE0000} N}  [O \ref{obs:MatchPol}]     \\ \hline
Strict                 & Arb.     & {\color[HTML]{FE0000} N} [T \ref{thm:MaxMaxminSymm}] & {\color[HTML]{FE0000} N} [T \ref{thm:MaxMaxminSymm}] & {\color[HTML]{FE0000} N} [T \ref{thm:MaxMaxminSymm}] & {\color[HTML]{FE0000} N} [T \ref{thm:MaxMaxminSymm}] & {\color[HTML]{FE0000} N} [T \ref{thm:MaxMaxminSymm}] & {\color[HTML]{FE0000} N} [T \ref{thm:MaxMaxminSymm}] & {\color[HTML]{FE0000} N} [T \ref{thm:EnvFree1DBUtil}]  & {\color[HTML]{FE0000} N}   [T \ref{thm:EfsUtility}]     & {\color[HTML]{FE0000} N} [T \ref{thm:EnvyWutilSymmStrict}]       & {\color[HTML]{FE0000} N}  [O \ref{obs:MatchPol}]        & {\color[HTML]{FE0000} N}  \cite{bodlaender_hedonic_2020, cechlarova_complexity_2002, cechlarova_exchange-stable_2005}     & {\color[HTML]{FE0000} N}  [O \ref{obs:MatchPol}]     \\ \hline
Symm.                  & Arb.     & {\color[HTML]{FE0000} N} [T \ref{thm:MaxMaxminSymm}] & {\color[HTML]{FE0000} N} \cite{bodlaender_hedonic_2020} & {\color[HTML]{FE0000} N} [T \ref{thm:MaxMaxminSymm}] & {\color[HTML]{FE0000} N} [T \ref{thm:MaxMaxminSymm}] & {\color[HTML]{FE0000} N} \cite{bodlaender_hedonic_2020}[T \ref{thm:MaxMaxminSymm}] & {\color[HTML]{FE0000} N} [T \ref{thm:MaxMaxminSymm}] & {\color[HTML]{FE0000} N} [T \ref{thm:EfbUtilityTies}] & {\color[HTML]{FE0000} N}   \cite{bodlaender_hedonic_2020, ceylan_optimal_2022}     & {\color[HTML]{FE0000} N}  [T \ref{thm:EnvyWutilSymmBinary}]      & {\color[HTML]{3166FF} P, } {\color[HTML]{32CB00}$\exists$} [T \ref{thm:exchStabBUtil}] & {\color[HTML]{32CB00} $\exists$} \cite{massand_graphical_2019} & {\color[HTML]{32CB00} $\exists$} [T \ref{thm:exchStabWUtil}] \\ \hline
Binary                  & Arb.     & {\color[HTML]{FE0000} N} [T \ref{thm:MaxMaxminSymm}] & {\color[HTML]{FE0000} N} \cite{bodlaender_hedonic_2020} & {\color[HTML]{FE0000} N} [T \ref{thm:MaxMaxminSymm}] & {\color[HTML]{FE0000} N} [T \ref{thm:MaxMaxminSymm}] & {\color[HTML]{FE0000} N} \cite{bodlaender_hedonic_2020}[T \ref{thm:MaxMaxminSymm}] & {\color[HTML]{FE0000} N} [T \ref{thm:MaxMaxminSymm}] & {\color[HTML]{FE0000} N} [T \ref{thm:EfbUtilityTies}] & {\color[HTML]{FE0000} N}   \cite{massand_graphical_2019, ceylan_optimal_2022}     & {\color[HTML]{FE0000} N}  [T \ref{thm:EnvyWutilSymmBinary}]      & Open & Open & Open \\ \hline
                       & Paths  &    {\color[HTML]{FE0000} N} [P \ref{prop:MaxMaxMinBUtil1D}]                        &        Open &    Open &    {\color[HTML]{FE0000} N} [P \ref{prop:MaxMaxMinBUtil1D}]                        &     Open  &  Open &  {\color[HTML]{FE0000} N} [T \ref{thm:EnvFree1DBUtil}]                            & {\color[HTML]{680100} $\nexists$} [T \ref{thm:EnvFree1dNotExists}] &  Open   & {\color[HTML]{3166FF} P, } {\color[HTML]{32CB00}$\exists$} [T \ref{thm:exchStabBUtil}] & {\color[HTML]{3166FF} P, } {\color[HTML]{32CB00}$\exists$} [P\ref{prop:ExchStab1D}] & {\color[HTML]{3166FF} P, } {\color[HTML]{32CB00}$\exists$} [P\ref{prop:ExchStab1D}] \\ \cline{2-14} 
\multirow{-2}{*}{1-D}  & Cycles &  {\color[HTML]{FE0000} N} [P \ref{prop:MaxMaxMinBUtil1D}]                          &  Open  &   Open &  {\color[HTML]{FE0000} N} [P \ref{prop:MaxMaxMinBUtil1D}]                          &  Open &   Open    &   {\color[HTML]{FE0000} N} [T \ref{thm:EnvFree1DBUtil}]                         & {\color[HTML]{680100} $\nexists$} [T \ref{thm:EnvFree1dNotExists}] & Open & {\color[HTML]{3166FF} P, } {\color[HTML]{32CB00}$\exists$} [T \ref{thm:exchStabBUtil}] & {\color[HTML]{3166FF} P, } {\color[HTML]{32CB00}$\exists$} [P\ref{prop:ExchStab1D}] & {\color[HTML]{3166FF} P, } {\color[HTML]{32CB00}$\exists$}  [P\ref{prop:ExchStab1D}]         \\ \hline
\end{tabular}
\end{adjustbox}
\caption{Summary of results for \textsc{seat arrangement}. Under each of the four problems (MWA, MUA,...), we have results for each of the three utility types. B stands for B-utility, S for S-utility and W for W-utility. Similarly, in the first column, Val stands for valuations, Arb means arbitrary seat graph, and Symm means symmetric valuations. In the second column Match means matching graph. Cells with {\color[HTML]{FE0000} N} indicate that the problem is NP-hard, while cells with {\color[HTML]{3166FF} P} indicate that the problem can be solved in polynomial time. In the case of EFA and STA, if a type of arrangement is guaranteed to exist, its cell in the table will contain {\color[HTML]{32CB00}$\exists$}. If the arrangement never exists, its cell will contain {\color[HTML]{680100} $\nexists$}. The results included in this paper are preceded by a letter, which indicates if the result is a theorem (T), an observation (O), or a proposition (P).}
\label{tab:resultSummary}
\end{table}


\subsection{Related work}
The \textsc{seat arrangement} problem was first considered by Bodlaender et al.\ \cite{bodlaender_hedonic_2020}. However, similar problems, such as \textsc{exchange-stable roommates} \cite{alcalde_exchange-proofness_1994, cechlarova_complexity_2002, cechlarova_exchange-stable_2005}, had already been studied prior to that work. In an \textsc{exchange-stable roommates} instance there is a set of agents, and the goal is to partition the set into subsets of size two (pairs) such that no two agents want to swap their partners, in which case we have an \textit{exchange-stable matching}. An application is the allocation of housing in college, in which each room houses two students. A matching is exchange-stable if no two students assigned to different rooms wish to switch their rooms, and hence their roommates. Thus, this problem is equivalent to determining the existence of an exchange-stable arrangement in \textsc{seat arrangement}, if we restrict ourselves to seat graphs where every connected component has two vertices. An important result in \textsc{exchange-stable roommates} is that an exchange-stable matching need not exist and deciding whether one does exist is NP-complete \cite{cechlarova_complexity_2002, cechlarova_exchange-stable_2005}. As \textsc{seat arrangement} is a generalization of \textsc{exchange-stable roommates}, this result holds for \textsc{seat arrangement} too. A related problem is \textsc{stable roommates} \cite{gale_college_1962}, in which the goal is to partition a set of agents into pairs, and a pair of agents blocks if they wish to be with each other instead of their respective partners in the matching. In contrast to the case for \textsc{exchange-stable roommates}, we can find a \textit{stable matching}, i.e., a matching with no blocking pairs, or report that none exists, in polynomial time \cite{gusfield_stable_1989}. Some other related problems include hedonic games \cite{brandt_handbook_2016}, Schelling games \cite{elkind_2019_schelling}, and topological distance games \cite{bullinger_2024_topological}.

Another important result for \textsc{seat arrangement} comes from Massand and Simon \cite{massand_graphical_2019}. In their paper they study a generalization of \textsc{seat arrangement} where agents derive an utility not just from the sum of their valuations over their neighbours, but also from an intrinsic valuation of the seat assigned to them. They show that, if preferences are symmetric, an exchange-stable arrangement always exists, and can be found by successively switching pairs of agents that desire to exchange seats, although the problem is PLS-complete. However, with arbitrary preferences, it was shown that finding an envy-free arrangement is NP-hard, even in the case where the valuation of every vertex is the same, meaning that EFA is hard in \textsc{seat arrangement} too.

Bodlaender et al.\ \cite{bodlaender_hedonic_2020} considered \textsc{seat arrangement} with S-utility, and showed that, for EFA, hardness remains even if preferences are symmetric and every graph component has at most three vertices. Likewise, MWA, the problem of finding an arrangement that maximises the total utility, is NP-hard, even for symmetric and binary preferences, for a large class of graphs. The same result applies to MUA, the problem of finding an arrangement that maximises the utility of the agent with the lowest utility, but only if Spanning Subgraph Isomorphism is NP-complete in regular graphs, which is currently unknown. On the other hand, if every component of the seat graph has two vertices, then EFA, MWA and MUA are solvable in polynomial time.

Due to the computational hardness of most problems related to \textsc{seat arrangement}, there has been some work on parameterized complexity, and in finding fixed-parameter tractable algorithms in particular \cite{cygan_parameterized_2015}. For instance, Bodlaender et al.\ \cite{bodlaender_hedonic_2020} studied the parameterized complexity of MWA, MUA, EFA, and STA where the parameter is the vertex cover number of the seat graph. They show, among other results, that MWA parameterized by the vertex cover number is W[1]-hard, and that both MUA and EFA are weakly NP-hard even if the vertex cover number is two. Ceylan et al.\ also gave parameterized complexity results for \textsc{seat arrangement} with S-utility \cite{ceylan_optimal_2023}. They considered two parameters, the number of non-isolated vertices in the seat graph, and the maximum number of non-zero valuations that any agent can give to other agents, and provide both FPT algorithms and W[1]-hardness results.

Another approach, which we also take on this work, is to restrict the preferences of agents. For instance, Berriaud et al. \cite{berriaud_stable_2023} study envy-freeness and exchange stability in instances with path and cycle graphs, where the agents can be partitioned into a number of classes, such that agents within one class are equivalent to each other. They find that these problems are solvable in polynomial time if the number of classes is bounded, but they remain NP-hard if they are not \cite{berriaud_stable_2023}. Similarly, Ceylan studied envy-freeness, and showed that EFA is NP-hard under S-utility, even with symmetric and binary preferences \cite{ceylan_optimal_2022}.

Another restriction of agents' preferences is $1$-dimensional preferences. However, with the exception of a paper by Wilczynski \cite{wilczynski_ordinal_2023}, $1$-dimensional preferences had not been studied before in the context of \textsc{seat arrangement}. Moreover, in that paper, $1$-dimensional preferences were only considered in the problem of finding popular arrangements, so no previous results for $1$-dimensional preferences and STA, EFA, MWA or MUA are known\footnote{Positive results with arbitrary or symmetric preferences extend to $1$-dimensional preferences.}. Nevertheless, $1$-dimensional preferences, and generalizations like single-peaked preferences or $d$-dimensional preferences (for $d \geq 1$), have been studied before in different matching and social choice problems. Some examples are \textsc{stable marriage} and \textsc{stable roommates} \cite{arkin_geometric_2009, bartholdi_stable_1986, bredereck_stable_2020, chen_et_al:LIPIcs.ESA.2022.36, wen_stable_2023}, voting theory \cite{brandt_handbook_2016, black1948rationale, brandt_2010_bypassing}, and resource allocation problems \cite{sprumont_1991_division}.

Up until now, in all papers related to \textsc{seat arrangement}, the only type of agent utility considered was S-utility. However, that is not the case in the related problem of coalition formation with hedonic preferences \cite{brandt_handbook_2016}, where the goal is to find a partition of a given set of agents that satisfies a concept of stability that depends on the agents' preferences. In such problems, the concept of agents only caring about their preferred, or least-preferred, agent in the partition is present in much of the literature \cite{cechlarova_computational_2003, cechlarova_stable_2004, cechlarova_stability_2001}. Determining the existence and finding such partitions is usually possible in polynomial time, but some variants are NP-hard. Coalition formation with hedonic preferences is related to \textsc{seat arrangement}, as an instance of the latter in which every connected component of the seat graph is a clique can be understood as an instance of coalition formation with hedonic preferences, as every connected component would be a distinct coalition. Nevertheless, results, positive or otherwise, from one problem do not translate to the other. The main difference is that, in coalition formation with hedonic preferences, a coalition of agents blocks \footnote{Other stability concepts do exist, see here \cite{brandt_handbook_2016} for an overview, but they are generally not equivalent to exchange-stability.} if each agent can break away from their assigned coalition and form a new coalition in which they are better off. However, this is not possible in \textsc{seat arrangement}, as an agent can only get a new seat if they swap seats with another agent. This point is also what differentiates \textsc{stable roommates} from \textsc{exchange-stable roommates}.

\subsection{Paper outline}
Section \ref{sec:Prel} gives definitions of necessary notation and terminology for this paper. Section \ref{sec:ExcStab} presents the results on STA, Section \ref{sec:EnvF} gives the results on EFA, and Section \ref{sec:MaxMaxmin} contains the results for MWA and MUA. Finally, Section \ref{sec:Conc} contains the conclusion.

\section{Preliminaries}
\label{sec:Prel}
A \textsc{seat arrangement} instance $I = (P,F,G,U)$ has four elements: a set $P$ of \textit{agents}, a set $F$ of \textit{valuation functions} $f_p: P \backslash \{p\} \rightarrow \mathbb{R}$ for each $p \in P$ that determines $p$'s preferences over all other agents, an undirected graph $G$, with as many vertices as there are agents and where the agents will be placed, called the \textit{seat graph}, and a \textit{utility type} $U$, which indicates the mechanism that is used to derive a utility for an agent depending on their placement in the seat graph. Our objective is to find an assignment from agents to vertices of $G$ that meets an optimality or fairness goal.

Let $p, q, r \in P$, $p \notin \{q,r\}$. The \textit{valuation} that $p$ gives to $q$ is the number $f_p(q)$. We say that $p$ \textit{strictly prefers} $q$ over $r$ if $f_p(q) > f_p(r)$, and is \textit{indifferent between} $q$ and $r$ if $f_p(q) = f_p(r)$.

We say that agents' preferences are:
\begin{itemize}
    \item \textit{nonnegative} if $\forall p,q \in P$, $p \neq q$, $f_p(q) \geq 0$;
    \item \textit{positive} if $\forall p,q \in P$, $p \neq q$, $f_p(q) > 0$;
    \item \textit{binary} if $\forall p,q \in P$, $p \neq q$, $f_p(q) \in \{0,1\}$;
    \item \textit{symmetric} if $\forall p,q \in P$, $p \neq q$, $f_p(q) = f_q(p)$;
    \item \textit{strict} if $\forall p,q,r \in P$, $p \notin \{q,r\}$, if $q \neq r$, then $f_p(q) \neq f_p(r)$;
    \item \textit{$1$-dimensional} if agents are points in the real line. Here, $d(p,q)$ represents the distance between $p$ and $q$ in the line, and $D$ denotes the distance between the leftmost and the rightmost agent. Then $\forall p \neq q \in P$, $f_p(q) = D - d(p,q) + 1$. $1$-dimensional preferences are, by definition, positive, because $D \geq d(p,q)$, and the $+1$ term ensures that $f_p(q) > 0$. Likewise, they are symmetric, as $f_p(q) = D - d(p,q) + 1 = D - d(q,p) + 1 = f_q(p)$. If no two agents are at the same point in the real line, their positions are \textit{unique}. We refer to the position or location of agent $a$ in the $1$-dimensional space as $l_a$, so if $l_a = 1$ then $a$ is at position $1$.
\end{itemize}
If agents' preferences are symmetric, we say that a pair of distinct agents $\{p, q\}$ has \textit{score} $s \in \mathbb{R}$ if $f_p(q) = f_q(p) = s$.

Let $G$ be the seat graph whose vertices will be assigned to agents from $P$. In every instance, the number of vertices in $G$ equals the number of agents in $P$. $G$ can be an arbitrary graph, but it can also belong to one of a number of possible graph classes. We say that $G$ is
\begin{itemize}
    \item a \textit{matching graph} if every connected component is a $K_2$, i.e., two vertices, joined by an edge;
    \item a \textit{path (cycle) graph} if every connected component consists of a path (cycle). A cycle with $k$ vertices will be referred to as a $k$-cycle;
    \item a \textit{cluster graph} if every connected component is a complete subgraph (a clique).
\end{itemize}

We define an \textit{arrangement} (alternatively an \textit{assignment}) to be a bijection $\sigma: P \rightarrow V(G)$. Thus $\sigma(p) \in G$ is the vertex in $G$ that agent $p$ is assigned to. Alternatively, $\sigma(p)$ is the seat of $p$ under $\sigma$. If agents' preferences are symmetric, we say that an edge $e \in E(G)$ has \textit{score} $s \in \mathbb{R}$ under $\sigma$ if $e = \{\sigma(p), \sigma(q)\}$ for $p,q \in P$, and the pair $\{p,q\}$ has score $s$. For any arrangement $\sigma$ and agents $p, q \in P$, $p \neq q$, let $\sigma^{p \leftrightarrow q}$ be the arrangement where

\begin{itemize}
    \item $\forall r \in P$, $r \notin \{p,q\}$, $\sigma^{p \leftrightarrow q}(r) = \sigma(r)$
    \item $\sigma^{p \leftrightarrow q}(p) = \sigma(q)$
    \item $\sigma^{p \leftrightarrow q}(q) = \sigma(p)$
\end{itemize}
Thus, $\sigma^{p \leftrightarrow q}$ is the same arrangement as $\sigma$ but with $p$ and $q$ swapping positions.

A function $u_{\sigma}: P \rightarrow \mathbb{R}$ assigns a real number to every agent, which represents the \textit{utility} that the agent derives under arrangement $\sigma$. We will consider three different types of utility functions. We say that $u_{\sigma}$ is an \textit{additively separable} utility function if $u_{\sigma}(p) = \sum_{v \in N(\sigma(p))} f_p(\sigma^{-1}(v))$, where $N(\sigma(p))$ is the set that includes all vertices adjacent to $\sigma(p)$. That is, the utility of $p$ under $\sigma$ is the sum of the valuations under $f_p$ for each of the agents neighbouring $p$. We will refer to additively separable utility functions as \textit{S-utility} functions also (S stands for sum). Alternatively, $u_{\sigma}$ is a \textit{B-utility} (respectively \textit{W-utility}) function if $u_{\sigma}(p) = \max_{v \in N(\sigma(p))} f_p(\sigma^{-1}(v))$ (resp. $\min_{v \in N(\sigma(p))} f_p(\sigma^{-1}(v))$). In other words, under B-utility, agent $p$ cares only about its best neighbour, while under W-utility, $p$ cares only about its worst neighbour. If an agent $p$ is assigned to a vertex with no neighbours, then its utility is $0$ under all utility types. In a \textsc{seat arrangement} instance $I = (P,F,G,U)$, $U$ specifies the type of all utility functions in the instance, with $U \in \{B,S,W\}$.

We will now consider solution concepts corresponding to optimal or fair arrangements. First, we define two specific types of fair arrangement. Under arrangement $\sigma$, we say that agent $p$ \textit{envies} agent $q$ if $u_{\sigma}(p) < u_{\sigma^{p \leftrightarrow q}}(p)$. Informally, $p$ envies $q$ if under arrangement $\sigma$ $p$ prefers to be seated where $q$ is seated. Agents $p$ and $q$ form an \textit{exchange-blocking} pair if they envy each other. An arrangement $\sigma$ is called \textit{envy-free} if no agent $p$ envies another agent $q$, and $\sigma$ is called \textit{exchange-stable} if no exchange-blocking pair exists. Envy-freeness implies exchange-stability.

We can now define formally the four different problems that will be studied.

\textsc{Maximum Welfare Arrangement (MWA)}. Find an arrangement $\sigma$ that maximises $\sum_{p \in P} u_{\sigma}(p)$, the sum of the utilities of all agents.

\textsc{Maximin Utility Arrangement (MUA)}. Find an arrangement $\sigma$ that maximises \\ $\min_{p \in P} u_{\sigma}(p)$, the utility of the agent with the smallest utility in $\sigma$.

\textsc{Stable Arrangement (STA)}. Find an exchange-stable arrangement, or report that none exists.

\textsc{Envy-Free Arrangement (EFA)}. Find an envy-free arrangement, or report that none exists.

Before presenting the main results, we introduce a small result that serves to adapt some previously known results for S-utility to B- and W-utility.

\begin{observation}
In a matching graph $G$, all three utility types considered in this paper are equivalent. Therefore, results for S-utility translate to B- and W-utility.
\label{obs:MatchPol}
\end{observation}
\begin{proof}
In a matching graph, the degree of any vertex is $1$. Thus, for a given agent $p$ and arrangement $\sigma$, $\max_{v \in N(\sigma(p))} f_p(\sigma^{-1}(v)) = \sum_{v \in N(\sigma(p))} f_p(\sigma^{-1}(v)) = \min_{v \in N(\sigma(p))} f_p(\sigma^{-1}(v))$, i.e., $p$ has the same utility under all three different utility types because there is only one neighbour.
\end{proof}

Therefore, for both B-utility and W-utility, finding an exchange-stable arrangement is NP-hard if preferences are strict but otherwise arbitrary, even if the seat graph is a matching graph, because under S-utility finding a exchange-stable arrangement is NP-hard if preferences are strict and the seat graph is a matching graph \cite{bodlaender_hedonic_2020, cechlarova_complexity_2002, cechlarova_exchange-stable_2005}. However, if the seat graph is a matching graph, EFA, MWA and MUA are solvable in polynomial time \cite{bodlaender_hedonic_2020}.

\section{Exchange stability (STA)}
\label{sec:ExcStab}
In this section we study the problem of finding an exchange-stable arrangement, or reporting that none exists. By previous work \cite{bodlaender_hedonic_2020, cechlarova_complexity_2002, cechlarova_exchange-stable_2005} and Observation \ref{obs:MatchPol}, we know that, for all three utility types, if agents' preferences are arbitrary, then determining the existence of an exchange-stable arrangement is NP-complete, even if the seat graph is a matching graph. By contrast, if agents' preferences are symmetric and the utility type is S-utility, an exchange-stable arrangement always exists, and can be found using the following process: starting from an arbitrary arrangement, while an exchange-blocking pair $\{p,q\}$ exists, switch the seats of $p$ and $q$, and continue until no such pair exists \cite{massand_graphical_2019, bodlaender_hedonic_2020}. In this case, every step of this process produces an arrangement with strictly higher total utility than the previous one, and by a potential function argument it converges on an exchange-stable arrangement. This potential function argument was first introduced by Bogomolnaia and Jackson \cite{bogomolnaia_2002_stability}, and has been used since several times \cite{bilo_hedonic_2022, bullinger_2024_topological, massand_graphical_2019, wilczynski_ordinal_2023}. We now show that this process of swapping agents belonging to exchange-blocking pairs also works for W-utility.

In order to do that, we need to introduce some new concepts. For an instance of \textsc{seat arrangement} with symmetric preferences, let $S$ be the set of different scores of pairs of agents. Given an assignment $\sigma$, $n(\sigma) = \langle t_1, t_2,... \rangle$ is a vector of size $|S|$, where $t_i$ equals the number of edges in the seat graph $G$ with score equal to the $i$-th highest element in $S$. Let $n(\sigma)[i]$ be $t_i$, the $i$-th element of $n(\sigma)$. We say that an arrangement $\sigma$ is \textit{w-better} than an arrangement $\sigma'$ if there exists $i$ ($1 \leq i \leq |S|$) such that $n(\sigma)[i] < n(\sigma')[i]$ and, $\forall j > i$, $n(\sigma)[j] = n(\sigma')[j]$.

\begin{theorem}
Let $I=(P,F,G,W)$ be an instance of \textsc{seat arrangement}, with symmetric agents' preferences and W-utility. Then an exchange-stable matching always exists, and we can arrive at one by satisfying exchange-blocking pairs.
\label{thm:exchStabWUtil}
\end{theorem}
\begin{proof}
Let $\sigma$ be an arrangement, with $q$ and $r$ forming an exchange-blocking pair. Let $\sigma(q) = u$ and $\sigma(r) = v$. Furthermore, let $G' \subset G$ be the subgraph of $G$ that contains every vertex and edge of $G$ except for $u$, $v$, and their adjacent edges. Let $\sigma_{G'}$ be the restriction of $\sigma$ to $G'$, so that for any agent $s \in P$, if $\sigma(s) \notin \{u,v\}$, then $\sigma(s) = \sigma_{G'}(s)$.

Now, let $\sigma^{q \leftrightarrow r}$ be the arrangement in which $q$ and $r$ switch positions with respect to $\sigma$. As the position of every other agent is the same, we have that $\sigma_{G'} = \sigma_{G'}^{q \leftrightarrow r}$, which means that $n(\sigma_{G'}) = n(\sigma_{G'}^{q \leftrightarrow r})$.

On the other hand, let $G''$ be the subgraph of $G$ that includes $u$, $v$, their neighbours, and every edge that includes one of $u$ or $v$. Thus, every edge of $G$ is included in exactly one of $G'$ and $G''$. As $q$ and $r$ form an exchange blocking pair, their utilities under $\sigma^{q \leftrightarrow r}$ are strictly higher than under $\sigma$. As the utility type is W-utility, this implies that the edge with the lowest score in the neighbourhood of $\sigma_{G''}^{q \leftrightarrow r}(q) = v$ has a strictly higher score than the edge with the lowest score in the neighbourhood of $\sigma_{G''}(q) = u$, and likewise for $r$. Hence, $n(\sigma_{G''}^{q \leftrightarrow r})$ is w-better than $n(\sigma_{G''})$. Hence $n(\sigma^{q \leftrightarrow r}) = n(\sigma_{G'}^{q \leftrightarrow r}) + n(\sigma_{G''}^{q \leftrightarrow r})$ is w-better than $n(\sigma) = n(\sigma_{G'}) + n(\sigma_{G''})$.

Therefore, every time we swap the seats of two agents who form an exchange-blocking pair, we get a strictly w-better arrangement. As the number of different arrangements is finite, this process eventually concludes, and the resulting arrangement is exchange-stable.
\end{proof}

However, there is no guarantee that this algorithm runs in polynomial time. In fact, STA with symmetric preferences and W-utility might also be PLS-complete, like STA with symmetric preferences and S-utility. We leave this as an open problem. As an aside, we should also note the similarity between our argument and that of Wilczynski \cite{wilczynski_ordinal_2023}, although their way of comparing profiles and our way work differently.

In the case of B-utility, if agents' preferences are symmetric, an exchange-stable arrangement always exists, by an argument similar to that of Theorem \ref{thm:exchStabWUtil}. Whilst the algorithm implied by Theorem \ref{thm:exchStabWUtil} does not necessarily run in polynomial time, the alternative approach that we now describe will give a polynomial-time algorithm to find an exchange-stable arrangement under B-utility if agents' preferences are symmetric.

We build an exchange-stable arrangement as follows. First, we order every pair of agents according to its score. Then, we select two agents $p_i, p_j \in P$ such that $\{p_i, p_j\}$ has the highest score, and assign them to neighbouring vertices $u, v$. Every agent that has $p_i$ as a favourite agent will be assigned to a neighbour of $u$, and similarly every agent that has $p_j$ as their favourite agent will be assigned to a neighbour of $v$, until either all vertices in $N(u)$ (respectively $N(v)$) have an agent assigned to them or all agents that preferred $p_i$ (respectively $p_j$) the most have been assigned. For all agents $p_k$ that have been assigned in the previous step, we allocate to the neighbours of $\sigma(p_k)$ any agent whose most preferred agent is $p_k$. We iterate this process with any new allocation, until we either run out of agents to assign or vertices to be assigned to.

If there are still two neighbouring vertices without agents assigned to them, we repeat the process outlined above, with a new pair $p_l, p_m \in P$ such that $\{p_l, p_m\}$ has the highest score among pairs made out of unassigned agents, and we keep doing this until there is no edge between two unassigned vertices. If there are agents still to be allocated, they can only be assigned to vertices where every neighbour has an agent allocated there. Thus, for every remaining agent $p$ and free vertex $v$ we can calculate the utility that $p$ would get if she is arranged into $v$. Therefore, the last step of the algorithm is to arbitrarily order the remaining agents, and assign each agent to a free vertex where she would get the highest utility among all free vertices. Algorithm \ref{alg:BUtilSTA} presents this procedure.

\begin{algorithm}[ht]
\caption{Algorithm to find an exchange-stable arrangement, when preferences are symmetric and the utility type is B-utility}
\label{alg:BUtilSTA}
\begin{algorithmic}
\Require Instance $I=(P,F,G,B)$ of Seat Arrangement
\Ensure Exchange-stable arrangement $\sigma$
\State $L = \emptyset$ \Comment{Set of agents allocated to $G$ with a free vertex as a neighbour}
\State $P' = P$
\State create set of ordered pairs $Q$

\Comment{While there are two vertices sharing an edge without an agent allocated to them}
\While{$\exists u,v \in V(G)$ such that $\{u,v\} \in E(G)$, $\sigma^{-1}(u) = \emptyset$, $\sigma^{-1}(v) = \emptyset$}

\Comment{Get a pair of agents $\{p,q\} \in Q$ with the highest score, with $p, q \in P'$}
\State $\{p,q\} = \arg \max_{\{s,t\} \in Q, s,t \in P'} score(\{s,t\})$
\State $\sigma(p) = u$, $\sigma(q) = v$
\State $P' = P' \setminus \{p,q\}$
\If{$\exists w \in N(u)$ such that $\sigma^{-1}(w) = \emptyset$}
\Comment{If there is a free vertex $w$ next to $u$}
    \State $L = L \cup \{p\}$
\EndIf
\If{$\exists w \in N(v)$ such that $\sigma^{-1}(w) = \emptyset$}
\Comment{If there is a free vertex $w$ next to $v$}
    \State $L = L \cup \{q\}$
\EndIf

\Comment{While there is an unallocated agent whose favourite agent is in $L$}
\While{$\exists r \in P'$ such that its favourite agent $s$ from $P' \cup L$ is in $L$}
\If{$\exists w \in N(\sigma(s))$ such that $\sigma^{-1}(w) = \emptyset$}
\Comment{If there is a free vertex $w$ next to $\sigma(s)$}
\State $\sigma(r) = w$
\State $P' = P' \setminus \{r\}$
\If{$\exists w' \in N(w)$ such that $\sigma^{-1}(w') = \emptyset$}
\Comment{If there is a free vertex $w'$ next to $w$}
    \State $L = L \cup \{r\}$
\EndIf
\Else
\Comment{If no free vertex next to $\sigma(s)$, remove $s$ from $L$}
    \State $L = L \setminus \{s\}$
\EndIf
\EndWhile
\EndWhile
\ForAll{$p \in P'$}
\Comment{We now deal with leftover agents}
    \State allocate $p$ to a free vertex in $G$ where it would get highest utility
\EndFor
\end{algorithmic}
\end{algorithm}

Lemma \ref{lem:Algorithm1TimeComplexity} establishes that Algorithm \ref{alg:BUtilSTA} runs in polynomial time, while Theorem \ref{thm:exchStabBUtil} establishes that Algorithm \ref{alg:BUtilSTA} produces an exchange-stable matching.

\begin{lemma}
Let $I=(P,F,G,B)$ be an instance of \textsc{seat arrangement} with symmetric preferences and B-utility. Then Algorithm \ref{alg:BUtilSTA} can be implemented to run in $O(n^3)$ time, where $n = |P|$, the number of agents in the instance. Under certain conditions, the algorithm can be executed in $O(n^2)$ time, implying that its complexity is linear in the size of the instance.
\label{lem:Algorithm1TimeComplexity}
\end{lemma}
\begin{proof}
We build three different data structures for representing the instance $I$: a list of pairs of agents $Q$, an array of agents $P$, and an array $G$ where each element is a vertex from the seat graph. We now describe each structure. First, $Q$ is a list of pairs of agents $\{p,q\}$. These pairs are ordered by their score, with ties between pairs that have the same score broken arbitrarily. Each pair $\{p,q\}$ stores the indices of agents $p$ and $q$ in $P$. Additionally, every pair that includes agent $p$ has a pointer to the next pair in the list that includes $p$. The pointer of the last such pair is null. Clearly these pointers can be created in time linear in the length of $Q$, i.e., $O(n^2)$ time.

Let us describe the array $P$ now. Every element in the array is an agent, which are ordered arbitrarily. Every agent $p$ has pointers to the previous and next agent in $P$ (so $P$ is also a doubly-linked list), and it also stores the index of the vertex $v$ that $p$ is arranged into under our arrangement $\sigma$. Initially, this value is null. Also, every agent $p$ has a list associated with it, and a pointer to access such list. This list, $F_p$ ($F$ from favourite), contains every agent $q$ such that $\{p,q\}$ is a pair with the highest score among all pairs that include $q$ and where the other agent is either unassigned or assigned to a vertex that has a free vertex as a neighbour (in the pseudocode above, the other agent is in $P'$ or $L$). That is, $F_p$, contains all agents that prefer $p$ the most among all eligible agents (initially all agents are eligible). Every element in $F_p$ stores the index of the original agent in $P$. Finally, for every agent $p$ we look at the pair $\{p, q\}$ with the highest score among all pairs including $p$, and count how many pairs including $p$ have the same score. Thus, this number, which we refer to as $c_p$, represents the number of joint-favourite agents that $p$ has. Every agent $p$ has one such $c_p$, as well as a pointer to a pair $\{p,r\} \in Q$ such that $\{p,r\}$ is the first pair that comes in $Q$ after all the pairs including $p$ with the highest score, that is, all pairs that we counted to arrive at $c_p$. In other words, $\{p,r\}$ is the next-best pair, and we store a pointer to it.

Finally, we represent the seat graph as an array of vertices $G$. Every vertex $v \in G$ stores the index of the agent $p$ assigned to $v$ (initially null), and has pointers to the previous and next vertices. Thus, $G$, like $P$, has a doubly-linked list on top of the array. Each vertex $v$ has an associated array that includes every vertex that $v$ is adjacent to. Every element in this array stores the index of its vertex in $G$, and every vertex $v$ has a pointer pointing to the first element of its adjacency array, which we call the iterating pointer of $v$. The array $G$ also has a pointer to the first vertex, which we will call the iterating pointer of $G$.

The algorithm has four steps that we need to analyse: creating all data structures, finding two adjacent vertices with no agent assigned to either, assigning a vertex to all agents that like a recently arranged agent $p$ the most, and dealing with the leftover agents. The first step takes $O(n^2 \log n)$ time, because there are $O(n^2)$ pairs of agents, and ordering them in order to create $Q$ takes $O(n^2 \log n)$ time. Building the other data structures, $P$ and $G$ (and the associated lists, like $F_p$ for every $p \in P$), takes $O(n^2)$ time.

Next comes the step of finding two free and adjacent vertices, and assigning agents to them. The first part, finding the two vertices, can be done in $O(n^2)$ time. We do so by, starting with the vertex $v$ that the iterating pointer of $G$ points at, checking if it is free. If it is, we go over the adjacency array, using the iterating pointer of $v$ to check if any adjacent vertex is free. Vertices that have an agent assigned to them do not become unassigned later, so these iterating pointers only need to advance in one direction over the entire execution of the algorithm, making this part take $O(n^2)$ time over the entire execution of the algorithm. The second part, finding two agents to arrange into the free vertices, also requires $O(n^2)$ time over the entire execution of the algorithm. This happens because we have to select a pair of unassigned agents with the highest score and assigned agents cannot become unassigned, so only one pass of $Q$ over the entire execution is needed, and, thanks to the different pointers and stored indices, checking if agents are assigned to vertices, and assigning them, can be done in constant time.

The third step consists of, for every agent $p$ assigned to a vertex $v$, assigning the agents that like $p$ the most, i.e., the agents in $F_p$, to vertices adjacent to $v$, until either all agents in $F_p$ are arranged into the graph or all vertices adjacent to $v$ have an agent assigned to them. For every $q \in F_p$, either $q$ has already been arranged into the graph, in which case we do nothing, $q$ has not been arranged yet but there is a vertex $u$ adjacent to $v$ with no assigned agent, in which case we arrange $q$ into $u$ (and we repeat this process for all agents in $F_q$), or $q$ has not been arranged but there is no such $u$. In this last case, we reduce $c_q$ by $1$. If, at some point of the execution, $c_q$ reaches $0$, it means that agent $q$ cannot be placed next to any of its most-preferred agents, as all of them have been assigned a vertex already and none of these vertices have an adjacent vertex with no agent assigned to it. In this case, we use the pointer from $q$ to the next-best pair $\{q,r\}$ and check all pairs with the same score that also include $q$. For each such $r$, we check if $r$ has been assigned a vertex $w_1$, and if there is a free vertex $w_2$ adjacent to $w_1$. In this case (which corresponds to $r \in L$ in the pseudocode), we assign $q$ to $w_2$. If all of the other agents in these pairs have been placed already and $q$ cannot be placed next to them, we move onto the next set of pairs. If these two cases do not hold, then at least some agent $s$ is in one of these pairs with $q$ and has not been assigned a vertex yet (in the pseudocode, such agent $s$ is in $P'$). In this case, we count all such agents, store that number in $c_q$, create a pointer to the first pair including $q$ after the pairs with the same score as $\{q,s\}$, and add $q$ to every $F_s$. Thanks to the use of the different pointers and data structures, for every agent $p$ arranged into the graph, notifying every $F_p$ and placing them if possible takes $O(n)$ time. Similarly, for every agent $p$, going down the list of pairs after $c_p$ reaches $0$ takes $O(n)$ time over the entire execution of the algorithm, because of the pointers linking the pairs that include $p$. Hence, in total, these two steps take $O(n^2)$ for all agents over the entire execution of the algorithm.

The final step occurs when there are no adjacent vertices that have no agents assigned to them. To make this last part more efficient, every time we assign an agent $p \in P$ to a vertex $v \in G$, we remove $p$ from the doubly-linked list that exists on top of the array $P$, and we do the same with $v$ and the doubly-linked list on top of $G$. This way, by the time we reach this step, these doubly-linked lists only contain the unassigned agents and the free vertices. Every unassigned agent $p$ must check every free vertex $v$. This process implies checking the entire adjacency list of $v$, which at this point only contains vertices with agents assigned to them, so that $p$ can calculate its utility when assigned to $v$. To do so, for every vertex $u$ adjacent to $v$, such that the agent assigned to $u$ is $q$, we must check the score of $\{p,q\}$. With an appropriate data structure, this last step can be done in constant time, so checking the entire adjacency list of vertex $v$ can be done in $O(n)$ time. If the number of unassigned agents is $k$, iterating over every remaining vertex for each agent can be completed in $O(kn)$ time. Thus, this last step takes $O(k^2n)$ time. As $k = O(n)$, dealing with the leftover agents takes $O(n^3)$ time, taken over the entire execution of the algorithm.

Putting this all together, the entire algorithm takes $O(n^2 \log n + n^2 + n^2 + n^3) = O(n^3)$ time. However, if $k$ is much smaller than $n$, then the last step takes $O(n)$ time. Similarly, if pairs are already ordered, the first step takes $O(n^2)$ time. Therefore, the algorithm may take $O(n^2)$ time on the number of agents. As the size of the instance is quadratic on the number of agents, the algorithm could be executed in linear time on the size of the instance, which means it is optimal if the conditions from this paragraph hold.
\end{proof}

\begin{theorem}
Let $I=(P,F,G,B)$ be an instance of \textsc{seat arrangement} with symmetric preferences and B-utility. Then Algorithm \ref{alg:BUtilSTA} produces an exchange-stable matching in polynomial time.
\label{thm:exchStabBUtil}
\end{theorem}
\begin{proof}
Let $\sigma$ be an arrangement for instance $I$ obtained as the output of Algorithm \ref{alg:BUtilSTA}. Let $p,q \in P$ be two different agents such that Algorithm \ref{alg:BUtilSTA} allocated a vertex to $p$ before allocating a vertex to $q$. Assume that $q$ envies $p$. Aiming for a contradiction, assume that $p$ envies $q$ too, which means that they form an exchange-blocking pair. We have that there are three possibilities for when $p$ was allocated a vertex.

The first possibility is that $p$ was chosen as one of the two agents in the pair with the highest score, across all remaining pairs. However, this means that $p$ preferred the other agent in the pair, $r$, over both agents that had not been allocated yet, and agents that were allocated and had a free vertex in their neighbourhood. Thus, if $p$ is envious of $q$, one of the neighbours of $q$ is an agent whose allocated vertex had no free neighbour at the time that $p$ was allocated. However, this is a contradiction, as $q$ was allocated after $p$, so if said vertex was free for $q$, it was free for $p$, too.

The second possibility is that one of the allocated agents $r$ had a vertex in its neighbourhood with no allocated agent, and $r$ was $p$'s most preferred agent in the subset of agents that, at that point in time, were either not allocated or had a free vertex in their neighbourhood. We arrive at a contradiction here because of similar reasons to those of the first possibility.

Finally, the third possibility is that $p$ was allocated at a point when no two free vertices shared an edge. At this point in the execution of the algorithm, for every free vertex, their entire neighbourhood had been allocated agents, so $p$ could calculate its utility in every free vertex, and choose the vertex where it could get a higher utility. As $p$ was allocated a vertex before $q$, it chose a vertex before $q$ did, so if it feels envy towards $q$, it means that it chose a vertex where its utility would not be maximised across all possibilities, a contradiction.

No matter at which point of the algorithm $p$ was allocated a vertex, as $p$ was allocated a vertex before $q$, if $p$ feels envy towards $q$ we get a contradiction. Hence, there is no pair of agents that feel envy towards one another, so the arrangement $\sigma$ is exchange-stable.
\end{proof}

The polynomial-time complexity of Algorithm \ref{alg:BUtilSTA} stands in contrast with the PLS-completeness of finding an exchange-stable arrangement under S-utility and symmetric preferences. In fact, similar problems to STA with symmetric preferences and S-utility are also PLS-complete. Some examples are in hedonic games \cite{gairing_2019_computing} and topological distance games \cite{bullinger_2024_topological}. Thus, it is of interest to study whether those other problems become polynomial-time solvable under notions equivalent to B-utility.

With symmetric preferences and either S-utility or W-utility, no algorithm is known that can find an exchange-stable arrangement in polynomial time. However, if agents' preferences are $1$-dimensional, we can always find an exchange-stable arrangement for either utility type, if the seat graph is a path graph or a cycle graph. We do so using the following procedure. In a path graph, take any connected component (a path) and allocate to one of its end vertices the leftmost agent in the $1$-dimensional space. Then we assign to the neighbour of that end vertex the leftmost unassigned agent, and we keep doing that until all vertices in the path have an agent allocated to them. Then we iterate with the remaining agents and connected components, until every agent has been allocated a vertex. Thus, this algorithm arranges agents consecutively. If the seat graph is a cycle graph, the procedure is identical, except that there are no end vertices, so the first agent to be assigned to a cycle is allocated to an arbitrary vertex within the cycle instead of an end vertex. This algorithm works in linear time. We now show that, for both S-utility and W-utility, it produces an exchange-stable arrangement.

\begin{proposition}
Let $I = (P,F,G,U)$ be an instance of \textsc{seat arrangement} with $1$-dimensional preferences, a path or cycle seat graph, and utility type either S-utility or W-utility. Then the procedure outlined above produces an exchange-stable arrangement in linear time in the number of agents.
\label{prop:ExchStab1D}
\end{proposition}
\begin{proof}
We will only prove the case with W-utility and a seat graph that is a path graph; the other cases can be proved in a similar manner. Let $\sigma$ be an arrangement produced using the procedure above, and, aiming for a contradiction, assume it is not exchange-stable. Thus, there exist agents $p,q \in P$ such that they prefer arrangement $\sigma^{p \leftrightarrow q}$ over $\sigma$. Without loss of generality, assume that, in the $1$-dimensional instance, $p$ is to the left of $q$, and that neither $p$ nor $q$ are assigned to an end vertex, so that they both have two neighbours (the other cases, where one or both of $p$ and $q$ are allocated to an end vertex, are proven in a similar manner). As agents are allocated consecutively, the two neighbours of $p$, $p_{-1}$ and $p_1$, are the agents immediately to its left and to its right, respectively, in the $1$-dimensional space. Similarly, the neighbours of $q$ under $\sigma$ are $q_{-1}$ and $q_1$, such $q_{-1}$ is immediately to the left of $q$ in the $1$-dimensional space, and $q_1$ is immediately to the right.

If $p$ wants to swap seats with $q$, it means that $\min(f_p(q_{-1}), f_p(q_1))$ is strictly higher than $\min(f_p(p_{-1}),$ $f_p(p_1))$, and thus $f_p(q_1)$ is strictly higher than $\min(f_p(p_{-1}), f_p(p_1))$. We have that $q_1$ is to the right of $p$, but $p_1$ is the agent immediately to the right of $p$, so $f_p(p_1) = D - d(p, p_1) + 1 \geq D - d(p, q_1) + 1 = f_p(q_1)$. Therefore, we get that $f_p(q_1) > f_p(p_{-1})$, and thus $d(p, q_1) < d(p, p_{-1})$. Likewise, $f_q(p_{-1}) > f_q(q_1)$, and thus $d(q, p_{-1}) < d(q, q_1)$. However, $d(p_{-1}, p) \leq d(q, p_{-1}) < d(q, q_1) \leq d(p, q_1) < d(p, p_{-1})$, a contradiction. Therefore, $\sigma$ is exchange-stable.
\end{proof}

\section{Envy-freeness (EFA)}
\label{sec:EnvF}
We now present results on finding envy-free arrangements or determining that none exists. For all utility types, EFA is in NP, as for any given arrangement we can check for every pair of agents $p,q$ whether $p$ envies $q$. We have seen that, generally, determining whether an instance of \textsc{seat arrangement} contains an envy-free arrangement is NP-complete, if the utility type is S-utility. An important exception occurs when the seat graph is a matching graph. In that case, EFA is solvable in polynomial time \cite{bodlaender_hedonic_2020}. In general, an envy-free arrangement is not guaranteed to exist, so each hardness result in this section implies that the related problem of determining whether an envy-free arrangement exists is NP-complete.

First, we show that EFA with B-utility is NP-hard. In order to prove this result, we are going to reduce from Partition Into Triangles. In this problem, we are given a graph $G=(V,E)$, where $|V| = 3n$ for some $n \in \mathbb{Z}$. The goal is to decide whether there exists a partition of $V$ into $3$-element subsets such that the elements of each subset form a triangle in $G$, where a triangle means that each vertex in the subset is a neighbour of every other vertex. This problem is NP-complete, even if $G$ has maximum degree $4$ \cite{caprara2002packing, van_rooij_partition_2013}.

\begin{theorem}
\textsc{Envy-Free Arrangement} under B-utility is NP-hard even if agents' preferences are symmetric and binary, the seat graph has bounded degree, and every graph component is either a cycle or a star.
\label{thm:EfbUtilityTies}
\end{theorem}
\begin{proof}
We reduce from Partition Into Triangles (PIT). Let $G = (V,E)$ be an instance from PIT with $3n$ vertices, such that every vertex has degree at least $2$ (otherwise no partition exists), and maximum degree $4$. We create an instance $I = (P,F,H,B)$ of the \textsc{seat arrangement} problem. For every vertex $v_i \in G$ we create an agent $p_i \in P$. Let $V_i$ be the set of neighbours of $v_i$. Then, for $v_j, v_k \in V_i$, we create agent $q_i^{j,k} \in P$, with $j < k$. This agent represents the possibility of vertex $v_i$ being in the same triangle as vertices $v_j$ and $v_k$. For vertex $v_i$, the set of agents of type $q_i^{j,k}$ is $Q_i \subseteq P$. As $G$ has maximum degree $4$, $v_i$ has at most four neighbours in $I$, and as there are six ways of selecting two elements out of a four-element subset, $|Q_i| \leq 6$. Finally, for each vertex $v_i$ we create several agents, $r_i^l \in P$, $1 \leq l \leq 2 \cdot |Q_i| + 2$, which belong to subset $R_i \subseteq P$.

For agents $p_i, q_i^{j,k} \in P$, $f_{p_i}(q_i^{j,k}) = f_{q_i^{j,k}}(p_i) = 1$. If vertices $v_i, v_j, v_k$ are all neighbours of each other, $f_{q_i^{j,k}}(q_j^{i,k}) = f_{q_j^{i,k}}(q_i^{j,k}) = 1$, and similarly $f_{q_i^{j,k}}(q_k^{i,j}) = f_{q_k^{i,j}}(q_i^{j,k}) = 1$ and $f_{q_j^{i,k}}(q_k^{i,j}) = f_{q_k^{i,j}}(q_j^{i,k}) = 1$. For agents $p_i, r_i^l$, $f_{p_i}(r_i^l) = f_{r_i^l}(p_i) = 1$. For every other pair of agents $a,b$, $f_a(b) = f_b(a) = 0$. Finally, the seat graph $H$ consists of $n$ cycles of size $3$, and $3n$ stars, one for each agent. The $i$-th star, corresponding to vertex $i$, has one centre vertex, and $|Q_i| + |R_i| - 1$ leaves. As $1 \leq |Q_i| \leq 6$ and $4 \leq |R_i| \leq 14$, every star has between four and nineteen leaves, so $H$ has bounded degree.

Assume there is a partition into triangles. Then if vertices $v_i, v_j, v_k$ are in one triangle, we assign to one of the cycles of size three agents $q_i^{j,k}, q_j^{i,k}, q_k^{i,j}$. We do the same for every other triangle. In the $i$-th star, we put agent $p_i$ in the middle, and we assign to the leaves the agents from $Q_i$ that are not in any cycle, as well as the agents from $R_i$. As exactly one agent from $Q_i$ is assigned to a vertex in a cycle, and there are $|Q_i| + |R_i| - 1$ leaves, the assignment is valid. It is also an envy-free arrangement, as every agent gets utility $1$, the maximum possible. An example of the transformation, with the envy-free arrangement included, is given in Figure \ref{fig:ExampleTheoremEnvyFreeBUtilityTies}.

\begin{figure}[ht]
    \centering
    \begin{subfigure}[b]{0.35\linewidth}
        \centering
        \begin{tikzpicture}[main/.style = {node distance={25mm}, thick, draw, circle}]
            \node[main] (1) {$v_1$};
            \node[main] (2) [above right of=1] {$v_2$};
            \node[main] (3) [below right of=1] {$v_3$};
            \node[main] (4) [above left of=1] {$v_4$};
            \draw[line width=1mm] (1) -- (2);
            \draw[line width=1mm] (1) -- (3);
            \draw[line width=1mm] (2) -- (3);
            \draw (1) -- (4);
            \draw (2) -- (4);
            \draw (3) to [out=180,in=270] (4);
        \end{tikzpicture}
        \vspace{3cm}
    \end{subfigure}
    \hspace{0.05\linewidth}
    \begin{subfigure}[b]{0.55\linewidth}
        \centering
        \begin{tikzpicture}[main/.style = {node distance={25mm}, thick, draw, circle}]
            \node[main] (1) at (1,5) {$q_1^{2,3}$};
            \node[main] (2) at (-1,5) {$q_2^{1,3}$};
            \node[main] (3) at (0,6.732) {$q_3^{1,2}$};
            \node[main] (4) at (0,0) {$p_1$};
            \node[main] (5) at (2,0) {$q_1^{2,4}$};
            \node[main] (6) at (1.246, 1.564) {$q_1^{3,4}$};
            \node[main] (7) at (-0.446, 1.95) {$r_1^1$};
            \node[main] (8) at (-1.802, 0.868) {$r_1^2$};
            \node[main] (9) at (-1.802, -0.868) {$r_1^3$};
            \node[main] (10) at (-0.446, -1.95) {$r_1^4$};
            \node[main] (11) at (1.246, -1.564) {$r_1^5$};
            \draw (1) -- (2);
            \draw (1) -- (3);
            \draw (2) -- (3);
            \draw (4) -- (5);
            \draw (4) -- (6);
            \draw (4) -- (7);
            \draw (4) -- (8);
            \draw (4) -- (9);
            \draw (4) -- (10);
            \draw (4) -- (11);
        \end{tikzpicture}
    \end{subfigure}
    \caption{On the left, subgraph of an example PIT instance, with one triangle from the partition highlighted. On the right, part of an envy-free arrangement, focused on the agents derived from vertex $v_1$.}
    \label{fig:ExampleTheoremEnvyFreeBUtilityTies}
\end{figure}

Now, assume there exists an envy-free arrangement $\sigma$ in $I$. We first show that this implies that every agent has utility $1$ in $\sigma$.

\begin{claim}
Let $\sigma$ be an envy-free arrangement of $I$. Then, for all $p \in P$, $u_{\sigma}(p) = 1$.
\end{claim}
\begin{proof}\renewcommand{\qedsymbol}{\ensuremath{\blacksquare}}
Assume otherwise. Then there exists some $a \in P$ such that $u_{\sigma}(a) = 0$. However, there exists an agent $b$ such that $f_a(b) = 1$. As there are no isolated vertices in $H$, $b$ has at least one neighbour $c$ under $\sigma$. Thus, with a B-utility function, $u_{\sigma^{a \leftrightarrow c}}(a) = 1$, so $a$ is envious of $c$, contradicting the envy-freeness of $\sigma$.
\end{proof}

Now we prove that if, in $\sigma$, every agent has utility $1$, there exists a partition into triangles in $G$.

\begin{claim}
Let $\sigma$ be an arrangement of $I$ such that, for all $p \in P$, $u_{\sigma}(p) = 1$. Then $G$ can be partitioned into triangles.
\end{claim}
\begin{proof}\renewcommand{\qedsymbol}{\ensuremath{\blacksquare}}
We first show that in an envy-free arrangement $\sigma$ agent $p_i$ is allocated to the centre vertex in the $i$-th star (or a star with the same number of vertices). We have that every star has at least four leaves. However, every agent that is not of type $p_x$ has a preference of $1$ for at most three other agents. Therefore, if such an agent $a$ is assigned to a centre vertex, there exists at least one agent $b$ assigned to a leaf vertex such that $f_a(b) = f_b(a) = 0$, a contradiction. Furthermore, $p_i$ cannot be assigned to the centre vertex in a star with more vertices than the $i$-th star. Assume otherwise and let $p_i$ be assigned to the centre vertex of the $j$-th star, such that the number of vertices of this star is larger than the number of vertices of the $i$-th star. Thus, $|Q_j| + |R_j| - 1 > |Q_i| + |R_i| - 1$. The way the sets are constructed, this inequality implies that $|Q_j| > |Q_i|$ and $|R_j| > |R_i|$, meaning the $j$-th star has at least two more vertices than the $i$-th star, and that $|Q_j| + |R_j| - 1 > |Q_i| + |R_i|$. Hence, if $p_i$ is assigned to the centre of the $j$-th star, then there exists at least one agent $a$ assigned to a leaf vertex that does not belong to $Q_i \cup R_i$, so $f_{p_i}(a) = f_a(p_i) = 0$, a contradiction. Similarly, $p_i$ cannot be assigned to the centre vertex in a star with fewer vertices than the $i$-th star, because this implies that some other $p_j$ is assigned to the centre vertex of a star with more vertices than the $j$-th, which we have seen would imply that $u_{\sigma}(a) = 0$ for some $a \in P$, a contradiction.

As $p_i$ must be assigned to the centre vertex of the $i$-th star, and all agents in $R_i$ have a preference of $1$ over $p_i$ and of $0$ over everyone else, these agents must be assigned to leaves in the $i$-th star. The remaining leaves must be occupied by agents from $Q_i$, as any other agent would have utility $0$. Hence, exactly one agent from $Q_i$ is not assigned to a vertex from the $i$-th star, or any other star. Let this agent be $q_i^{j,k}$. In order for its utility to be $1$ and thus not be envious of a neighbour of $p_i$, one of the two other vertices in the cycle it has been assigned to must be occupied by either $q_j^{i,k}$ or $q_k^{i,j}$. Without loss of generality, let this be $q_j^{i,k}$. If the third vertex is occupied by some agent $a \neq q_k^{i,j}$, its utility would be $0$. Hence, the third vertex must be assigned to $q_k^{i,j}$. That is, we have a triangle in the original instance with vertices $v_i, v_j, v_k$. As only one agent from $Q_i$ is in some cycle component, we form only one triangle with vertex $v_i$, meaning that we obtain a partition in $I$.
\end{proof}

Therefore, there exists an envy-free arrangement $\sigma$ in $I$ if and only if there exists a partition into triangles in $G$. As the second problem is NP-complete, so is the first one.
\end{proof}

If preferences are symmetric and strict, instead of binary and symmetric, EFA with B-utility is still NP-hard. In fact, even with $1$-dimensional and strict preferences, it is NP-complete to determine the existence of an envy-free arrangement. In order to prove this statement, we will reduce from the \textsc{Bin Packing} problem \cite{garey_computers_1979}. In an instance of \textsc{Bin Packing}, we are given a set $X$ of items, such that each item $x_i \in X$ has a size $s_i \in \mathbb{N}$, a capacity $C \in \mathbb{N}$ for each bin, and a number of bins $K$. The set of all item sizes is $S$. The goal is to decide whether $X$ can be partitioned into at most $K$ different subsets such that the sum of the sizes of the elements in each subset is not greater than $C$.

Given an instance of \textsc{Bin Packing}, if $\sum_{x_i \in X} s_i < C \cdot K$, then we will add $C \cdot K - \sum_{x_i \in X} s_i$ elements of size $1$ each. Clearly this modified instance has a bin packing if and only if the previous instance has a bin packing too, because we can use the extra elements to ``fill" any bin that is not full. Then, in order to ensure that the size of all elements is at least $2$, we will multiply every size by $2$, as well as the bin capacity $C$. These two changes are done to ensure that the transformed \textsc{seat arrangement} instance is a valid instance where the number of agents equals the number of vertices, and to ensure that from each item $x_i \in X$ we can obtain at least $2$ agents.

\begin{theorem}
Let $I = (P,F,G,B)$ be a \textsc{seat arrangement} instance with B-utility, $1$-dimensional strict preferences with unique positions, and a seat graph in which every connected component is a path with the same number of vertices. Then \textsc{Envy-Free Arrangement} is NP-hard.
\label{thm:EnvFree1DBUtil}
\end{theorem}
\begin{proof}
Let $(X,S,C,K)$ be an instance of \textsc{Bin Packing}, where $X$ is the set of elements, $C$ the bin capacity, and $K$ the target number of bins. The minimum size of every item is $2$, and $\sum_{x_i \in X} s_i = C \cdot K$. The bins in the \textsc{Bin Packing} will be labelled $C_1,\dots,C_K$. Construct the \textsc{seat arrangement} instance $I = (P,F,G,B)$ as follows. Take item $x_1 \in X$ and create $s_1$ agents, $p_1^j\in P$ ($1 \leq j \leq s_1$). All of these agents belong to subset $P_1 \subset P$. Lay them out in the real line starting with $p_1^1$, which will be placed at $0$. Then $p_1^2$ will be placed a distance of exactly $1$ to the right of $p_1^1$, $p_1^3$ will be placed a distance of $1 + \varepsilon$ to the right of $p_1^2$, where $\varepsilon$ is a suitable small positive number, $p_1^4$ will be placed a distance of $1 + 2\cdot \varepsilon$ to the right of $p_1^3$, and so on, until $p_1^{s_1}$, which will be placed a distance of $1 + (s_1 - 2) \cdot \varepsilon$ to the right of $p_1^{s_1-1}$. After completing the placement of all agents associated with item $x_1$, we then proceed with the placement of item $x_2$. Its first agent, $p_2^1 \in P$, will be placed a distance of $2$ to the right of $p_1^{s_1}$. Then $p_2^2$ is placed at a distance of $1$ to the right of $p_2^1$, $p_2^3$ is placed at a distance of $1 + \varepsilon$ to the right of $p_2^2$, and so on. We repeat this process for all $x_i \in X$, in increasing subscript order. The seat graph $G$ consists of $K$ paths with $C$ vertices each, labelled $V^1$ to $V^K$. Figure \ref{fig:ExampleTheoremEnvyFree1DimensionalBUtility} below shows an example of how two items from the \textsc{Bin Packing} are transformed into agents of the \textsc{seat arrangement} instance $I$.

\begin{figure}[h]
    \centering
    \begin{tikzpicture}[main/.style = {node distance={25mm}, thick, draw, circle}]
        \tkzInit[xmin=-1,ymin=-1,xmax=13,ymax=1]

        \node[main] (1) at (0,0) {$p_1^1$};
        \node[main] (2) at (2,0) {$p_1^2$};
        \node[main] (3) at (4,0) {$p_1^3$};
        \node[main] (4) at (6,0) {$p_1^4$};
        \node[main] (5) at (10,0) {$p_2^1$};
        \node[main] (6) at (12,0) {$p_2^2$};
        \draw (1) -- node[above] {\fontsize{9}{10}\selectfont $1$} (2);
        \draw (2) -- node[above] {\fontsize{9}{10}\selectfont $1+\varepsilon$} (3);
        \draw (3) -- node[above] {\fontsize{9}{10}\selectfont $1+2\varepsilon$} (4);
        \draw (4) -- node[above] {\fontsize{9}{10}\selectfont $2$} (5);
        \draw (5) -- node[above] {\fontsize{9}{10}\selectfont $1$} (6);
    \end{tikzpicture}
\caption{Location of the agents transformed from $x_1$ and $x_2$ in the $1$-dimensional space, with $s_1 = 4$ and $s_2 = 2$}
\label{fig:ExampleTheoremEnvyFree1DimensionalBUtility}
\end{figure}

Given the way preferences are constructed, for any $p_i^l$, $2 \leq l \leq s_i$, its unique most preferred agent is $p_i^{l-1}$, while for $p_i^1$ its unique most preferred agent is $p_i^2$. Thus, in an envy-free arrangement $\sigma$, $p_i^1$ is next to $p_i^2$, who is also next to $p_i^3$, and so on.

Now suppose that there is a packing of the items into $K$ bins. For each $\ell$ ($1 \leq \ell \leq K$), suppose that the $\ell$-th bin contains items $x_{j_1},\dots,x_{j_r}$ for some $r\geq 1$, where $j_1<j_2<\dots<j_r$. Place the agents in $P_{j_1},P_{j_2},\dots,P_{j_r}$ onto the path $V^\ell$ in left-to-right order that respects the placement of these agents on the real line as constructed above. As every item $x_i$ is packed into a bin, all agents from $P_i$ are assigned a vertex in $G$ such that they feel no envy towards any other agent.

Conversely, suppose that there is an envy-free arrangement of the agents onto $G$. Let $k$ ($1\leq k\leq K$) be given and suppose that path $V^k$ contains an agent $p_i^r$. Then by the envy-free property, $V^k$ contains every agent in $P_i$. It follows that the packing obtained by placing item $x_i$ into bin $C_k$ whenever the family $P_i$ is placed on path $V^k$ is a legal packing of all items into $K$ bins. Therefore, as \textsc{Bin Packing} is NP-complete, so is \textsc{seat arrangement} with B-utility, $1$-dimensional strict preferences, and a seat graph in which every connected component is a path with the same number of vertices.
\end{proof}

As $1$-dimensional preferences, by definition, are symmetric, this means that the problem of finding an envy-free arrangement when preferences are strict and symmetric, or reporting that none exists, is NP-hard, if the utility type is B-utility. Another corollary is that the result also applies if the seat graph is a cycle graph, or a cluster graph. In the proof above, if we change each component in the seat graph for a cycle, or a clique, it is still the case that an envy-free arrangement exists if and only if the elements in $X$ can be packed into $K$ bins.

Likewise, finding an envy-free arrangement is NP-hard when the \textsc{seat arrangement} instance $I$ has utility type S-utility. By Bodlaender et al.\ \cite{bodlaender_hedonic_2020}, this result holds even for symmetric preferences, by Massand and Simon \cite{massand_graphical_2019} the result holds even for binary preferences, and by Ceylan \cite{ceylan_optimal_2022} the result holds even for symmetric and binary preferences. We show that NP-hardness holds even for symmetric and strict preferences.

\begin{theorem}
\textsc{Envy-Free Arrangement} under S-utility is NP-hard even for symmetric and strict preferences and a seat graph with bounded degree where every graph component is either a cycle or a star.
\label{thm:EfsUtility}
\end{theorem}
\begin{proof}
To prove this result, we will adapt the reduction from Theorem \ref{thm:EfbUtilityTies}. The agent set $P$, and the seat graph $H$, and preference functions are identical. Regarding preferences, we start with the preference functions from the proof of Theorem \ref{thm:EfbUtilityTies}. We then define new preferences $f_a'(b)$ for each pair of agents $a, b \in P$ using the preferences $f_a(b)$ defined in the proof of Theorem \ref{thm:EfbUtilityTies} and a set of mutually different infinitesimaly small real numbers $\varepsilon_i$ for $i = 1, \dots,|P|^2$ as follows: for each pair of agents $a, b$ we choose a different $\varepsilon_i$ and set $f_a'(b) = f_a(b) + \varepsilon_i$.

Just like in the previous case, if $G$ can be partitioned into triangles we can build an envy-free arrangement, and if there is an envy-free arrangement every agent has utility at least $1$, as per Claim 1 from the proof of Theorem \ref{thm:EfbUtilityTies}, and Claim 2 from the same proof still applies, with minor modifications to its proof.
\end{proof}

Similarly, in an instance with W-utility, finding an envy-free arrangement or reporting that none exists is NP-hard, even with symmetric and binary preferences, or symmetric and strict preferences, as we show with the following results.

\begin{theorem}
\textsc{Envy-Free Arrangement} under W-utility is NP-hard even for symmetric and binary preferences and a seat graph with maximum degree $2$ where every graph component is a $2$-cycle or a $3$-cycle.
\label{thm:EnvyWutilSymmBinary}
\end{theorem}
\begin{proof}
We reduce from Partition Into Triangles (PIT). Let $G=(V,E)$ be an instance from PIT with $3n$ vertices ($n \geq 2$), such that every vertex has degree at least $2$ (otherwise no partition exists). We create an instance $I$ of the \textsc{seat arrangement} problem. For every vertex $v_i$ we create an agent $p_i \in P$. For agents $p_i, p_j \in P$, $f_{p_i}(p_j) = f_{p_j}(p_i) = 1$ if $v_i$ and $v_j$ are neighbours in $G$, otherwise $f_{p_i}(p_j) = f_{p_j}(p_i) = 0$. Additionally, for each $p_i \in P$, we create two extra agents, $p_i^1$ and $p_i^2$. The preferences of these agents are as follows: $f_{p_i}(p_i^1) = f_{p_i^1}(p_i) = 1$, and $f_{p_i^1}(p_i^2) = f_{p_i^2}(p_i^1) = 1$. Finally, for $1 \leq i \leq 3n$, $f_{p_i^2}(p_{i+1}^1) = f_{p_{i+1}^1}(p_i^2) = 1$, where $i+1$ is taken modulo $3n$. For any other pair of agents $a,b \in P$, $f_{a}(b) = f_{b}(a) = 0$. That is, every $p_i$ gives a symmetrical valuation of $1$ to $p_j$ if $v_i$ and $v_j$ are neighbours, as well as to $p_i^1$, while $p_i^1$ gives a symmetrical valuation of $1$ to $p_i^2$. Finally, every $p_i^2$ gives a symmetrical valuation of $1$ to the ``next'' $p_{i+1}^1$. All other valuations are $0$. The seat graph $H$ contains $n$ $3$-cycles and $3n$ $2$-cycles, one for each $i$, $1 \leq i \leq 3n$. We will use the $p_i^1$ and $p_i^2$ agents to force agents of type $p_i$ to be arranged into $3$-cycles in any envy-free arrangement.

Assume there is a partition into triangles in $G$. Then we assign to each $3$-cycle in $H$ the three agents from each triangle. For each $1 \leq i \leq 3n$ we assign to a $2$-cycle agents $p_i^1$ and $p_i^2$. This is an envy-free arrangement, because every agent is getting the maximum utility possible, namely $1$. An example of such an envy-free arrangement is given in Figure \ref{fig:ExampleTheoremEnvyFreeWUtilityTies}.

\begin{figure}[ht]
    \centering
    \begin{subfigure}[b]{0.35\linewidth}
        \centering
        \begin{tikzpicture}[main/.style = {node distance={25mm}, thick, draw, circle}]
            \node[main] (1) {$v_1$};
            \node[main] (2) [above right of=1] {$v_2$};
            \node[main] (3) [below right of=1] {$v_3$};
            \node[main] (4) [above left of=1] {$v_4$};
            \draw[line width=1mm] (1) -- (2);
            \draw[line width=1mm] (1) -- (3);
            \draw[line width=1mm] (2) -- (3);
            \draw (1) -- (4);
        \end{tikzpicture}
        \vspace{2cm}
    \end{subfigure}
    \hspace{0.05\linewidth}
    \begin{subfigure}[b]{0.55\linewidth}
        \centering
        \begin{tikzpicture}[main/.style = {node distance={25mm}, thick, draw, circle}]
            \node[main] (1) at (1,5) {$p_1$};
            \node[main] (2) at (-1,5) {$p_2$};
            \node[main] (3) at (0,6.732) {$p_3$};
            \node[main] (4) at (-1,3) {$p_1^1$};
            \node[main] (5) at (1,3) {$p_1^2$};
            \node[main] (6) at (-1,1) {$p_2^1$};
            \node[main] (7) at (1,1) {$p_2^2$};
            \node[main] (8) at (-1,-1) {$p_3^1$};
            \node[main] (9) at (1,-1) {$p_3^2$};
            \draw (1) -- (2);
            \draw (1) -- (3);
            \draw (2) -- (3);
            \draw (4) -- (5);
            \draw (6) -- (7);
            \draw (8) -- (9);
        \end{tikzpicture}
    \end{subfigure}
    \caption{On the left, subgraph of an example PIT instance, with one triangle from the partition highlighted. On the right, part of an envy-free arrangement.}
    \label{fig:ExampleTheoremEnvyFreeWUtilityTies}
\end{figure}

Now, assume there exists an envy-free arrangement $\sigma$ in $I$. Then at least one agent of type $p_i^1$ is arranged into a $2$-cycle. Assume otherwise. Then all agents of type $p_i^1$ occupy all vertices in the $3$-cycles, because there are $3n$ agents of type $p_i^1$ and $n$ $3$-cycles. Thus, all agents of type $p_i^2$ are in $2$-cycles. In this case, however, every agent of type $p_i^1$ will have utility $0$. For each of these $p_i^1$, her respective $p_i^2$ will be in a $2$-cycle, so she will feel envy towards the other agent in said $2$-cycle, a contradiction. Thus, at least some $p_i^1$ is in a $2$-cycle. Fix this $p_i^1$. As $p_i^2$ only gives a valuation of $1$ to $p_i^1$ and $p_{i+1}^1$, she must be placed in a $2$-cycle with either of them, because if $p_i^2$ is placed in a different $2$-cycle or in a $3$-cycle, her utility will be $0$ and she will feel envy towards the agent that is placed in the same cycle as $p_i^1$. Assume first that $p_i^2$ is in the same cycle as $p_i^1$ (the other case, where $p_i^2$ and $p_{i+1}^1$ are in the same $2$-cycle, works analogously). Then, as $f_{p_{i-1}^2}(p_i^1) = 1$, $p_{i-1}^2$ must have utility $1$ in $\sigma$, or otherwise she will envy $p_i^1$. The only other agent that $p_{i-1}^2$ gives a valuation of $1$ to is $p_{i-1}^1$, so these two agents must be placed in the same $2$-cycle. But that means that $p_{i-2}^2$ and $p_{i-2}^1$ must be in the same $2$-cycle too. Therefore, by extending this argument, for every pair of agents $p_j^1$ and $p_j^2$, they are allocated to the same $2$-cycle under $\sigma$. If, instead, $p_i^2$ and $p_{i+1}^1$ are in the same $2$-cycle, we get that every pair of agents $p_j^2$ and $p_{j+1}^1$ are in the same $2$-cycle. As there are $6n$ agents of type $p_i^1$ and $p_i^2$, and $3n$ of these $2$-cycles, it means that no agent of type $p_i$ is arranged into a $2$-cycle, and so all are arranged into the $3$-cycles.

For every agent $p_i$, we have that $u_{\sigma^{p_i \leftrightarrow p_i^2}}(p_i) = 1$, because $f_{p_i}(p_i^1) = 1$. Hence, as $\sigma$ is envy-free, $u_{\sigma}(p_i) = 1$, which means that, if $p_j$ and $p_k$ are the neighbours of $p_i$, $f_{p_i}(p_j) = 1$, and $f_{p_i}(p_k) = 1$. Thus, in the Partition Into Triangles instance $G$, $v_j$ and $v_k$ are in the neighbourhood of $v_i$. As this is the case for every $p_i$, we can get a triangle in $G$ from every $3$-cycle, and because every agent in these $3$-cycles represents a different vertex, each vertex from $G$ is in a unique triangle, meaning that we can partition $G$ into triangles.

Therefore, there exists an envy-free arrangement $\sigma$ in our \textsc{seat arrangement} instance $I$ if and only if there exists a partition into triangles in $G$. As the second problem is NP-complete, so is the first one.
\end{proof}

\begin{theorem}
\textsc{Envy-Free Arrangement} under W-utility is NP-hard even for symmetric and strict preferences and a seat graph where every connected component is a cycle.
\label{thm:EnvyWutilSymmStrict}
\end{theorem}
\begin{proof}
We reduce from Partition Into Triangles (PIT). Let $G = (V,E)$ be an instance of PIT with $3n$ vertices, such that every vertex has degree at least $2$ (otherwise no partition exists), and maximum degree $4$. We create an instance $I = (P,F,H,W)$ of the \textsc{seat arrangement} problem. For every vertex $v_i \in G$, let $V_i$ be the set of neighbours of $v_i$. Then, for $v_j, v_k \in V_i$, we create agent $p_{i,1}^{j,k} \in P$, with $j < k$. This agent represents the possibility of vertex $v_i$ being in the same triangle as vertices $v_j$ and $v_k$ and hence is created only when $v_i$, $v_j$ and $v_k$ are all adjacent to each other. We also create agents $p_{i,2}^{j,k}$ and $p_{i,3}^{j,k}$. Every agent of the form $p_{i,r}^{j,k}$ belongs to subset $P_i \subset P$. As $G$ has maximum degree $4$, $v_i$ has at most four neighbours in $G$, and as there are six ways of selecting two elements out of a four-element subset, $|P_i| \leq 6 \cdot 3 = 18$. If at this point $|P_i| < 18$, because the degree of $v_i$ is less than $4$, we create extra agents until $|P_i| = 18$. That is, if $|P_i| = 18 - l$ for some $l \geq 1$, we create $l$ additional agents $q_i^1, q_i^2,..., q_i^l$, which are then added to $P_i$. Finally, we create six extra agents, $s_1$, $s_2$, and $s_3$, and $t_1$, $t_2$, and $t_3$, which will be used to ``enforce'' that an envy-free arrangement corresponds to a partition into triangles. Let $S = \{s_1, s_2, s_3\}$ and $T = \{t_1, t_2, t_3\}$.

If vertices $v_i, v_j, v_k$ are all neighbours of each other, $f_{p_{i,1}^{j,k}}(p_{j,1}^{i,k}) = f_{p_{j,1}^{i,k}}(p_{i,1}^{j,k}) = 3$. Similarly $f_{p_{i,1}^{j,k}}(p_{k,1}^{i,j}) = f_{p_{k,1}^{i,j}}(p_{i,1}^{j,k}) = 3$ and $f_{p_{j,1}^{i,k}}(p_{k,1}^{i,j}) = f_{p_{k,1}^{i,j}}(p_{j,1}^{i,k}) = 3$.

To define valuations for agents within the same $P_i$ subset, we need an order relation $\prec_{P_i}$. Within each $P_i$ subset, $p_{i,r}^{j,k} \prec_{P_i} p_{i,c}^{\alpha,\beta}$ if $j<\alpha$, or $j=\alpha$ and $k<\beta$, or $j=\alpha$, $k=\beta$ and $r < c$. Likewise, $p_{i,r}^{j,k} \prec_{P_i} q_i^l$ for any $j,k,r,l$, and $q_i^j \prec_{P_i} q_i^k$ if $j < k$. Thus, for example, $p_{i,1}^{j,k} \prec_{P_i} p_{i,2}^{j,k} \prec_{P_i} p_{i,3}^{j,k} \prec_{P_i} p_{i,1}^{j,m} \prec_{P_i} p_{i,2}^{j,m} \prec_{P_i} p_{i,3}^{j,m} \prec_{P_i} \dots \prec_{P_i} q_i^1 \prec_{P_i} \dots \prec_{P_i} q_i^l$. Any subset $S \subseteq P_i$ also has the order relation $\prec_S$, which is inherited from $P_i$.

Let us now define a relation between agents from any subset $S \subseteq P_i$. For $a,b \in S$, $a \prec_S' b$ if $a \prec_S b$ and $\nexists c \in S$ such that $a \prec_S c \prec_S b$. If $a \prec_S' b$, we say that $a$ is the \textit{immediate predecessor} of $b$ in $S$, and $b$ is the \textit{immediate successor} of $a$ in $S$. Finally, to ensure that every agent has an immediate predecessor and an immediate successor, the last element of $S$, according to $\prec_S$, is the immediate predecessor in $S$ of the first element according to $\prec_S$, and the first element is the immediate successor in $S$ of the last element. This means that we can arrange the agents of $P_i$, or any subset $S$, in a circle, such that every agent is adjacent to her immediate predecessor and her immediate successor.

We now define valuations for agents within the same $P_i$ subset. If $a \prec_{P_i}' b$, then $f_a(b) = f_b(a) = 2$. For example, $f_{p_{i,1}^{j,k}}(p_{i,2}^{j,k}) = f_{p_{i,2}^{j,k}}(p_{i,1}^{j,k}) = 2$, $f_{p_{i,2}^{j,k}}(p_{i,3}^{j,k}) = f_{p_{i,3}^{j,k}}(p_{i,2}^{j,k}) = 2$ and $f_{q_i^j}(q_i^{j+1}) = f_{q_i^{j+1}}(q_i^j) = 2$. Finally, because in an envy-free arrangement one agent from $P_i$ is assigned to a $3$-cycle, and the remaining agents to the same $17$-cycle (as we are going to see below), we define valuations for some pairs of agents within $P_i$. First, let $Q_i \subseteq P_i$ be a subset of $P_i$ that contains every agent from $P_i$ except those of the form $p_{i,1}^{j,k}$. We assume that $\prec_{Q_i}$ is inherited from $P_i$. For every agent of the form $p_{i,2}^{j,k}$, let $a \in Q_i$ be her immediate predecessor in $Q_i$. This agent $a$ is either of the form $p_{i,3}^{\alpha,\beta}$ or $q_i^l$, and so is not the immediate predecessor of $p_{i,2}^{j,k}$ in $P_i$. Then $f_{p_{i,2}^{j,k}}(a) = f_a(p_{i,2}^{j,k}) = 1.5$. An intuition behind these valuations is that we want agents to give a high valuation towards their immediate predecessor and successor, because those are going to be her neighbours in an envy-free arrangement. Because in an envy-free arrangement one agent of the form $p_{i,1}^{j,k}$ is not assigned to the same cycle as the other agents from $P_i$, we need to ensure that some agents within $P_i$ give a slightly smaller valuation to agents that are their immediate predecessors or successors after removing said $p_{i,1}^{j,k}$. We refer to Figure \ref{fig:ExampleTheoremEnvyFreeWUtilityStrict} for an illustration of this ordering and these valuations.

Let us now define the valuations of agents in $S$ and $T$. For agents $s_i, s_j$, $f_{s_i}(s_j) = f_{s_j}(s_i) = 3$, and similarly for agents $t_i, t_j$, $f_{t_i}(t_j) = f_{t_j}(t_i) = 3$, where in both cases $1 \leq i,j \leq 3$. For agents $s_i, t_j$, $1 \leq i,j \leq 3$, $f_{s_i}(t_j) = f_{t_j}(s_i) = 2$. Finally, for agents $s_i, t_j, a$, $1 \leq i,j \leq 3$ and $a \notin S \cup T$, $f_{s_i}(a) = f_{a}(s_i) = 1$, and $f_{t_j}(a) = f_{a}(t_j) = -1$. For every other pair of agents $a,b$, $f_a(b) = f_b(a) = 0$. These valuations are not strict, so in order to make them strict we arbitrarily break ties in valuations, in a manner similar to the technique used in the proof of Theorem \ref{thm:EfsUtility}. Thus, to each valuation we add arbitrarily small, positive, and pairwise distinct $\varepsilon_i$, where $\varepsilon_i$ is defined as in the proof of Theorem \ref{thm:EfsUtility}.

The seat graph $H$ consists of $n+2$ cycles of size $3$, and $3n$ cycles of size $17$, one for each vertex in the PIT instance.

Assume there is a partition into triangles in $G$. Then if vertices $v_i, v_j, v_k$ are in one triangle, we assign agents $p_{i,1}^{j,k}, p_{j,1}^{i,k}, p_{k,1}^{i,j}$ to one of the $3$-cycles in $H$. We do the same for every other triangle in $G$. In the remaining two $3$-cycles, we assign $s_1$, $s_2$ and $s_3$ to one of them, and $t_1$, $t_2$ and $t_3$ to the other. Then, for each $17$-cycle we assign to its seats all agents from a given $P_i$ that are not assigned to some $3$-cycle, such that every agent is adjacent to her immediate predecessor and successor within $P_i \backslash \{p_{i,1}^{j,k}\}$, where $p_{i,1}^{j,k}$ is an agent assigned to a $3$-cycle which indicates that $v_i$, $v_j$ and $v_k$ form a triangle.

\begin{figure}[ht]
    \centering
    \begin{subfigure}[b]{0.3\linewidth}
        \centering
        \begin{tikzpicture}[main/.style = {node distance={25mm}, thick, draw, circle}]
            \node[main] (1) {$v_1$};
            \node[main] (2) [above right of=1] {$v_2$};
            \node[main] (3) [below right of=1] {$v_3$};
            \node[main] (4) [above left of=1] {$v_4$};
            \draw[line width=1mm] (1) -- (2);
            \draw[line width=1mm] (1) -- (3);
            \draw[line width=1mm] (2) -- (3);
            \draw (1) -- (4);
            \draw (2) -- (4);
            \draw (3) to [out=180,in=270] (4);
        \end{tikzpicture}
        \vspace{2cm}
    \end{subfigure}
    \hspace{0.05\linewidth}
    \begin{subfigure}[b]{0.6\linewidth}
        \centering
        \begin{tikzpicture}[main/.style = {node distance={25mm}, thick, draw, circle}]
            \node[main] (1) at (1,6) {$p_{1,1}^{2,3}$};
            \node[main] (2) at (-1,6) {$p_{2,1}^{1,3}$};
            \node[main] (3) at (0,7.732) {$p_{3,1}^{1,2}$};
            \node[main] (v0) at ({360/17 * 0}:4) {$p_{1,2}^{2,3}$};
            \node[main] (v1) at ({360/17 * 1}:4) {$p_{1,3}^{2,3}$};
            \node[main] (v2) at ({360/17 * 2}:4) {$p_{1,1}^{2,4}$};
            \node[main] (v3) at ({360/17 * 3}:4) {$p_{1,2}^{2,4}$};
            \node[main] (v4) at ({360/17 * 4}:4) {$p_{1,3}^{2,4}$};
            \node[main] (v5) at ({360/17 * 5}:4) {$p_{1,1}^{3,4}$};
            \node[main] (v6) at ({360/17 * 6}:4) {$p_{1,2}^{3,4}$};
            \node[main] (v7) at ({360/17 * 7}:4) {$p_{1,3}^{3,4}$};
            \node[main] (v8) at ({360/17 * 8}:4) {$q_1^1$};
            \node[main] (v9) at ({360/17 * 9}:4) {$q_1^2$};
            \node[main] (v10) at ({360/17 * 10}:4) {$q_1^3$};
            \node[main] (v11) at ({360/17 * 11}:4) {$q_1^4$};
            \node[main] (v12) at ({360/17 * 12}:4) {$q_1^5$};
            \node[main] (v13) at ({360/17 * 13}:4) {$q_1^6$};
            \node[main] (v14) at ({360/17 * 14}:4) {$q_1^7$};
            \node[main] (v15) at ({360/17 * 15}:4) {$q_1^8$};
            \node[main] (v16) at ({360/17 * 16}:4) {$q_1^9$};
            \draw (1) -- (2);
            \draw (1) -- (3);
            \draw (2) -- (3);
            \draw (v0) -- (v1);
            \draw (v1) -- (v2);
            \draw (v2) -- (v3);
            \draw (v3) -- (v4);
            \draw (v4) -- (v5);
            \draw (v5) -- (v6);
            \draw (v6) -- (v7);
            \draw (v7) -- (v8);
            \draw (v8) -- (v9);
            \draw (v9) -- (v10);
            \draw (v10) -- (v11);
            \draw (v11) -- (v12);
            \draw (v12) -- (v13);
            \draw (v13) -- (v14);
            \draw (v14) -- (v15);
            \draw (v15) -- (v16);
            \draw (v16) -- (v0);

        \end{tikzpicture}
    \end{subfigure}
    \caption{On the left, subgraph of an example PIT instance, with one triangle from the partition highlighted. On the right, part of an envy-free arrangement. We can see that in the $17$-cycle every agent from $P_i \backslash \{p_{1,1}^{2,3}\}$ is adjacent to their immediate predecessor and successor in $P_i \backslash \{p_{1,1}^{2,3}\}$.}
    \label{fig:ExampleTheoremEnvyFreeWUtilityStrict}
\end{figure}

An example of part of this arrangement is given in Figure \ref{fig:ExampleTheoremEnvyFreeWUtilityStrict}. This arrangement, which we call $\sigma$, is valid, because for each $P_i$, exactly one agent is assigned to a vertex in a $3$-cycle and the other $17$ agents are arranged into the same $17$-cycle. It is also an envy-free arrangement, as we now show.

\begin{lemma}
The arrangement $\sigma$ obtained from the partition into triangles in $G$ is envy-free.
\end{lemma}
\begin{proof}\renewcommand{\qedsymbol}{\ensuremath{\blacksquare}}
Under $\sigma$, agents from $S$ (respectively $T$) have utility at least $3$. If they swap seats with an agent in the same cycle, their utility is unchanged, because they have the same set of neighbours. If they swap seats with an agent from a different graph component, at least one of their neighbours will not be in $S$ (respectively $T$), which means that they will give a valuation strictly smaller than $3$ to that agent, which under W-utility means that their utility is less than $3$.

Similarly, agents belonging to some $P_i$ who are arranged into a $3$-cycle have utility at least $3$. If they swap seats with an agent from the same cycle their utility is unchanged, and if they swap seats with any other agent, they will give a valuation smaller than $3$ to at least one of their neighbours, and so their utility would be less than $3$.

Let us now consider agents assigned to a $17$-cycle. For every such agent $a$, either (i) $a$ is arranged next to two agents to which $a$ gives a valuation of at least $2$, or (ii) to an agent to which $a$ gives a valuation of at least $2$ and another agent to which $a$ gives a valuation of at least $1.5$. In case (i), if $a \neq p_{i,1}^{j,k}$ for any $j,k$, then $a$ is adjacent to the two agents that $a$ gives a highest valuation to, which means that if $a$ were to swap seats her utility would decrease. If $a = p_{i,1}^{j,k}$ for some $j,k$, then there are two agents that $a$ gives a valuation of at least $3$ to, which are $p_{j,1}^{i,k}$ and $p_{k,1}^{i,j}$. However, these two agents are in two $17$-cycles such that $a$ gives a valuation of less than $1$ to every other agent in these cycles. Hence, if $a$ were to swap seats her utility would decrease.

Finally, consider case (ii), where $a$ is assigned to a $17$-cycle and gives a valuation of at least $2$ to one neighbour and a valuation between $1.5$ and $2$ to another. In this case $a = q_i^l$, $a = p_{i,2}^{j,k}$ or $a = p_{i,3}^{j,k}$, for some $l,j,k$, and there exists an agent $p_{i,1}^{\alpha,\beta}$ such that $f_a(p_{i,1}^{\alpha,\beta}) \geq 2$. However, $p_{i,1}^{\alpha,\beta}$ is assigned to some $3$-cycle, and $a$ gives a valuation of less than $1$ to the other two agents in the $3$-cycle, so her utility would decrease if she swaps seats to be adjacent to $p_{i,1}^{\alpha,\beta}$. Similarly, the utility of $a$ would decrease if she were to swap seats with any other agent, because she would give a valuation of less than $1.5$ to at least one of her neighbours.
\end{proof}

Conversely, assume there exists an envy-free arrangement $\sigma$ in $I$. We are going to show that this implies that agents in $S$ belong to one $3$-cycle, and agents in $T$ belong to another one. First, assume that some agent $a \notin S \cup T$ has in her neighbourhood some $t_i \in T$. Then $u_{\sigma}(a) < 0$. As the seat graph $H$ has more than three connected components, $a$ can switch seats with an agent $b$ allocated to a connected component with no agents from $T$, so $a$ would envy $b$. Hence, no agent from $T$ has in her neighbourhood such an $a$, which means that for any $t_i$, her two neighbours are in $S \cup T$. Now, assume that some $t_i$ is arranged into a $17$-cycle, and her two neighbours are from $S \cup T$. Then, if either of these neighbours are adjacent to an agent not in $S$ or $T$, said neighbour would envy $t_i$. Therefore, the neighbours of $t_i$ are adjacent to agents in $S \cup T$ too, which means that there are at least $5$ agents from $S \cup T$ arranged together in the same $17$-cycle. As $|S \cup T| = 6$, there are either $5$ or all $6$ agents from $S \cup T$ arranged together in the same $17$-cycle. However, this implies that two agents from $S \cup T$ have a neighbour not in $S \cup T$, and thus they envy $t_i$, a contradiction to the envy-freeness of $\sigma$.

Hence, all agents in $T$ are arranged into $3$-cycles that only include agents in $S \cup T$. Because $|S \cup T| = 6$, either all agents in $T$ belong to the same cycle, or some $t_i$ is in a cycle with two agents from $S$, and the other two agents from $T$ are in a cycle with some $s_j$. However, in this case $t_i$ would envy $s_j$. Thus, all agents from $T$ belong to the same $3$-cycle. Now, if some $s_i \in S$ has in her neighbourhood an agent $a \notin S \cup T$, then $u_{\sigma}(s_i) < 2$, and thus she will feel envy towards any agent $t_j \in T$, because $u_{\sigma^{s_i \leftrightarrow t_j}}(s_i) \geq 2$. Hence, all agents in $S$ are allocated to the same $3$-cycle.

We now show that in $\sigma$, the agents assigned to each $17$-cycle belong to the same subset $P_i \subset P$. For a contradiction, assume the opposite. Then there are two agents, $a \in P_i$ and $b \in P_j$, that belong to different subsets and are adjacent to each other. However, because $f_a(b) = f_b(a) < 1$, $u_{\sigma}(a) < 1$. As $u_{\sigma^{a \leftrightarrow s_1}}(a) \geq 1$, $a$ envies $s_1$, a contradiction to the envy-freeness of $\sigma$. Therefore, all agents arranged into a $17$-cycle belong to the same $P_i$. As there are $3n$ $17$-cycles and $3n$ $P_i$'s, for every $P_i$, seventeen agents belonging to it are assigned to some $17$-cycle.

For every $P_i$, $|P_i| = 18$, so exactly one agent from every $P_i$ is not assigned to a vertex from a $17$-cycle, and thus is assigned to a $3$-cycle. If this agent is not of type $p_{i,1}^{j,k}$, then she gives a valuation of less than $1$ to every agent outside of $S \cup P_i$, and so her utility in $\sigma$ is $0$, which means that she envies $s_1$, a contradiction. Hence, the only agent from $P_i$ not assigned to a $17$-cycle is of type $p_{i,1}^{j,k}$. In order for her utility to be at least $1$ and thus not be envious of $s_1$, the two other vertices in the cycle she has been assigned to must be occupied by $p_{j,1}^{i,k}$ and $p_{k,1}^{i,j}$, all of which have utility at least $3$ under $\sigma$. That is, we have a triangle in the original instance with vertices $v_i, v_j, v_k$. As only one agent from $P_i$ is in some cycle component, we form only one triangle with vertex $v_i$, meaning that we obtain a partition into triangles in $G$.
\end{proof}

For S-utility, finding an envy-free arrangement or reporting that none exists is generally NP-hard. Surprisingly, if we restrict ourselves to the case where agents' preferences are $1$-dimensional, it can be the case that no envy-free arrangement can possibly exist, if the seat graph belongs to certain graph classes. We later show this in Theorem \ref{thm:EnvFree1dNotExists}, but, in order to prove it, we need the following lemma.

\begin{lemma}
Let the seat graph of a \textsc{seat arrangement} instance with S-utility be a cycle with at least four vertices. Let agents' preferences be $1$-dimensional, with unique positions. Then for every arrangement $\sigma$ an agent will envy one of its neighbours under $\sigma$.
\label{lem:Ef1dcycle}
\end{lemma}
\begin{proof}
We will prove the lemma by induction. The base case is an instance of \textsc{seat arrangement} $I = (P,F,G,S)$, where $P = \{p,q,r,s\}$, agents' preferences are derived from a $1$-dimensional space where $l_p < l_q < l_r < l_s$, and $G$ consists of a cycle with $4$ vertices, labelled $v_1$ to $v_4$, such that $v_i$ is a neighbour of $v_{i+1}$ ($i \mod 4$). For this part of the proof, $(x,y,z,w)$ means that agent $x$ is assigned to $v_1$, $y$ to $v_2$, $z$ to $v_3$, and $w$ to $v_4$. Without loss of generality, let us consider every arrangement where $p$ is assigned to $v_1$ (so arrangements of the type $(a,x,y,z)$). In arrangement

\begin{enumerate}
    \item $\sigma = (p,q,r,s)$, $p$ envies $q$, because $u_{\sigma^{p \leftrightarrow q}}(p) = f_p(q) + f_p(r) > f_p(q) + f_p(s) = u_{\sigma}(p)$
    \item $\sigma = (p,q,s,r)$, $q$ envies $p$, because $u_{\sigma^{q \leftrightarrow p}}(q) = f_q(r) + f_q(p) > f_q(p) + f_q(s) = u_{\sigma}(q)$
    \item $\sigma = (p,r,q,s)$, $p$ envies $r$, because $u_{\sigma^{p \leftrightarrow r}}(p) = f_p(q) + f_p(r) > f_p(r) + f_p(s) = u_{\sigma}(p)$
    \item $\sigma = (p,r,s,q)$, $q$ envies $p$, because $u_{\sigma^{q \leftrightarrow p}}(q) = f_q(r) + f_q(p) > f_q(p) + f_q(s) = u_{\sigma}(q)$
    \item $\sigma = (p,s,q,r)$, $p$ envies $r$, because $u_{\sigma^{p \leftrightarrow r}}(p) = f_p(q) + f_p(r) > f_p(r) + f_p(s) = u_{\sigma}(p)$
    \item $\sigma = (p,s,r,q)$, $p$ envies $q$, because $u_{\sigma^{p \leftrightarrow q}}(p) = f_p(q) + f_p(r) > f_p(q) + f_p(s) = u_{\sigma}(p)$
\end{enumerate}

All other arrangements, where agent $p$ is assigned to $v_2$, $v_3$, or $v_4$, are equivalent to these six.

Now, assume that if the cycle from the seat graph has $k$ vertices, an agent will envy one of its neighbours under any arrangement. Consider the case with $k+1$ vertices. The $k+1$ agents about to be assigned to these vertices are the original $k$ agents plus an extra agent $z$, which is located to the right of all other agents in the $1$-dimensional space from which preferences are derived.

Take any assignment $\sigma$ on all agents but $z$, on the original cycle with $k$ vertices. By the assumptions, there is at least one agent $q$ who is envious of one of her neighbours $r$. Let $p$ be the other neighbour of $q$, and $s$ the other neighbour of $r$. Let $\sigma'$ be an arrangement of the agents onto a cycle with $k+1$ vertices, in which every agent has the same two neighbours as in $\sigma$, except for two neighbouring agents $a$ and $b$, who in $\sigma'$ have $z$ between them. That is, we take $\sigma$, and between $a$ and $b$ we introduce a new vertex which is assigned to $z$. If, in $\sigma'$, we place $z$ anywhere except between $p$ and $q$, between $q$ and $r$, or between $r$ and $s$, then $q$ will still be envious of $r$, because $u_{\sigma}(q) = u_{\sigma'}(q)$, which is smaller than $u_{\sigma^{q \leftrightarrow r}}(q) = u_{\sigma'^{q \leftrightarrow r}}(q)$. We now look at each of these three cases separately.

\textbf{Case 1:} let $z$ be arranged between $r$ and $s$ in $\sigma'$, with no agent envying a neighbour. The positions of the agents in $G$ are shown is Figure \ref{fig:EnvyFree1DimSUtilCase1}. As $q$ was envious of $r$ and is not anymore, we have that $f_q(s) + f_q(r) > f_q(p) + f_q(r) \geq f_q(r) + f_q(z)$, so $d(q,s) < d(p,q) \leq d(q,z)$.
As $r$ is not envious of $q$ in $\sigma'$, $f_r(p) + f_r(q) \leq f_r(q) + f_r(z)$, from which we get that $d(r,z) \leq d(p,r)$. As $z$ is the rightmost agent, this inequality implies that $r$ is to the right of $p$. Likewise, as $r$ is not envious of $z$, $f_r(z) + f_r(s) \leq f_r(q) + f_r(z)$ and thus $d(q,r) \leq d(r,s)$, so $r$ is closer to $q$ than to $s$. $z$ is not envious of $r$, so $f_z(q) + f_z(r) \leq f_z(r) + f_z(s)$ and thus $d(z,s) \leq d(q,z)$. Hence $s$ is to the right of $q$ in the $1$-dimensional space. Also, because $q$ is closer to $s$ than to $p$, and $s$ is to the right of $q$, $q$ is to the right of $p$. Finally, $r$ is between $p$ and $s$, because $r$ is to the right of $p$, $s$ is to the right of $q$, and $r$ is closer to $q$ than to $s$, implying than $r$ is to the left of $s$.

Therefore, the possible ways in which these agents are ordered in the $1$-dimensional space are $l_p < l_r < l_q < l_s < l_z$, and $l_p < l_q < l_r < l_s < l_z$.

However, with the first order we get a contradiction, as $d(r,z) \leq d(p,r)$ implies $d(q,z) < d(p,q)$, but we know that $d(p,q) \leq d(q,z)$. Assuming the second order is true, if the cycle has $5$ vertices, $s$ is next to $p$, so $s$ is envious of $z$, as $f_s(z) + f_s(p) < f_s(z) + f_s(r)$. Hence, let $l$ be an agent next to $s$ in the seat graph. If $l$ is to the left of $r$ in the $1$-dimensional space then $f_d(z) + f_d(l) < f_d(z) + f_d(r)$, so $s$ would envy $z$. Thus, $l$ is to the right of $r$. However, in this case, $z$ envies $s$, as $f_z(r) + f_z(s) < f_z(s) + f_z(l)$. Thus, in $\sigma'$ an agent will always envy a neighbour.

\begin{figure}[ht]
    \centering
    \begin{tikzpicture}[main/.style = {node distance={25mm}, thick, draw, circle}]
        \node[main] (p) at (-1.5,0) {$p$};
        \node[main] (q) at (-1.5/1.41,1.5/1.41) {$q$};
        \node[main] (r) at (0,1.5) {$r$};
        \node[main] (z) at (1.5/1.41,1.5/1.41) {$z$};
        \node[main] (s) at (1.5,0) {$s$};
        \draw (p) -- (q);
        \draw (q) -- (r);
        \draw (r) -- (z);
        \draw (z) -- (s);
    \end{tikzpicture}
    \caption{Positions of the agents in part of the seat graph $G$. Case $1$.}
    \label{fig:EnvyFree1DimSUtilCase1}
\end{figure}

\textbf{Case 2:} let $z$ be between $q$ and $r$ in $\sigma'$. The positions of the agents in $G$ are shown is Figure \ref{fig:EnvyFree1DimSUtilCase2}. As $q$ was envious of $r$ in $\sigma$, $f_q(r) + f_q(p) < f_q(r) + f_q(s)$, so $d(q,s) < d(p,q)$. If $z$ is not to be envious of $q$ in $\sigma'$, then $r$ is to the right of $p$ in the $1$-dimensional space, as otherwise $f_z(p) > f_z(r)$. Similarly, $q$ is to the right of $s$, or $z$ would envy $r$ under $\sigma'$. If $q$ is to not be envious of $z$, then $d(p,q) \leq d(q,r)$. Similarly, $d(r,s) \leq d(q,r)$, or $r$ would envy $z$. If, in the $1$-dimensional space, $q$ is to the left of $p$, or between $p$ and $r$, then, as $q$ is to the right of $s$, $d(r,s) > d(q,r)$, a contradiction. Thus, $q$ is to the right of $r$, and because $r$ is to the right of $p$, $d(p,q) > d(q,r)$. However, we stated before that $d(p,q) \leq d(q,r)$, or else $q$ is envious of $z$ in $\sigma'$, a contradiction. Thus, regardless of the positions of the agents in the $1$-dimensional space, under $\sigma'$ at least one agent will envy one of its neighbours.

\begin{figure}[ht]
    \centering
    \begin{tikzpicture}[main/.style = {node distance={25mm}, thick, draw, circle}]
        \node[main] (p) at (-1.5,0) {$p$};
        \node[main] (q) at (-1.5/1.41,1.5/1.41) {$q$};
        \node[main] (z) at (0,1.5) {$z$};
        \node[main] (r) at (1.5/1.41,1.5/1.41) {$r$};
        \node[main] (s) at (1.5,0) {$s$};
        \draw (p) -- (q);
        \draw (q) -- (z);
        \draw (z) -- (r);
        \draw (r) -- (s);
    \end{tikzpicture}
    \caption{Positions of the agents in part of the seat graph $G$. Case $2$.}
    \label{fig:EnvyFree1DimSUtilCase2}
\end{figure}

\textbf{Case 3:} assume that $z$ is placed between $p$ and $q$ in $\sigma'$ The positions of the agents in $G$ are shown is Figure \ref{fig:EnvyFree1DimSUtilCase1}. As $q$ does not envy $r$ anymore, $f_q(z) + f_q(r) \geq f_q(r) + f_q(s)$, so $f_q(z) \geq f_q(s)$. As $z$ is the rightmost agent, $q$ must be to the right of $s$ in the $1$-dimensional space. We now consider the position of $p$. First, as $q$ used to envy $r$, we have that $f_q(p) < f_q(s)$. Thus, $p$ cannot be to the right of $q$, because in that case $f_q(p) > f_q(z) \geq f_q(s) > f_q(p)$. Similarly, she cannot be between $q$ and $s$, because that implies that $f_q(p) < f_q(s)$, a contradiction. Hence, $p$ must be to the left of $s$. $r$ is to the right of $p$, too, or else $f_r(p) + f_r(q) > f_r(q) + f_r(s)$, which would mean that in $\sigma$ $r$ envies $q$, which by symmetry with the first case implies that $\sigma'$ is not envy-free. Thus, we have that $f_z(p) + f_z(q) < f_z(q) + f_z(r)$, so $z$ is envious of $b$.

\begin{figure}[ht]
    \centering
    \begin{tikzpicture}[main/.style = {node distance={25mm}, thick, draw, circle}]
        \node[main] (p) at (-1.5,0) {$p$};
        \node[main] (z) at (-1.5/1.41,1.5/1.41) {$z$};
        \node[main] (q) at (0,1.5) {$q$};
        \node[main] (r) at (1.5/1.41,1.5/1.41) {$r$};
        \node[main] (s) at (1.5,0) {$s$};
        \draw (p) -- (z);
        \draw (z) -- (q);
        \draw (q) -- (r);
        \draw (r) -- (s);
    \end{tikzpicture}
    \caption{Positions of the agents in part of the seat graph $G$. Case $3$.}
    \label{fig:EnvyFree1DimSUtilCase3}
\end{figure}

In all cases, if in every assignment with $k$ agents one agent envies a neighbour, in every assignment with $k+1$ agents that is the case as well. Therefore, by induction, no envy-free assignment exists in cycles with $4$ or more agents.
\end{proof}

We use the previous lemma in the proof of the following theorem, which shows the impossibility of finding an envy-free arrangement if agents' preferences are $1$-dimensional, the utility type is S-utility, and one of the connected components of the seat graph is a path or a cycle.

\begin{theorem}
Let $I = (P,F,G,S)$ be an instance of \textsc{seat arrangement} with S-utility, $1$-dimensional preferences with unique positions, and a seat graph containing at least one path component with at least three vertices, or a cycle component with at least four. Then no envy-free assignment exists.
\label{thm:EnvFree1dNotExists}
\end{theorem}
\begin{proof}
We first prove the result for paths. Let $a, b, c \in P$ be agents in a \textsc{seat arrangement} instance with $1$-dimensional preferences with unique positions, and let $\sigma$ be an arrangement that assigns $a$ to the end vertex of a path component of the seat graph, $b$ to the neighbour of the vertex assigned to $a$, and $c$ to the other neighbour of the vertex assigned to $b$. As preferences are strictly positive, and the utility type is S-utility, $u_{\sigma^{a \leftrightarrow b}}(a) = f_a(b) + f_a(c) > f_a(b) = u_{\sigma}(a)$. Hence, $a$ envies $b$.

Regarding the result for cycles, let $I = (P,F,G,S)$ be an instance of \textsc{seat arrangement}, such that preferences are derived from a $1$-dimensional space, and $G$ contains a cycle component with four vertices or more. Let $\sigma$ be any arrangement. Then, by Lemma \ref{lem:Ef1dcycle}, an agent assigned to the cycle component with $4$ or more vertices envies another agent assigned to the same cycle, so $\sigma$ is not envy-free. As $\sigma$ was arbitrary, $I$ does not have an envy-free arrangement.
\end{proof}

In the proof of the result for path graphs in Theorem \ref{thm:EnvFree1dNotExists} we only used the fact that $1$-dimensional preferences are positive. Thus, we get the following corollary.

\begin{corollary}
Let $I = (P,F,G,S)$ be an instance of \textsc{seat arrangement} with S-utility, positive preferences, and a seat graph containing at least one path component with at least three vertices. Then $I$ has no envy-free arrangement.
\end{corollary}

\section{Maximum (MWA) and maximin (MUA) assignments}
\label{sec:MaxMaxmin}
We now show our results on MWA and MUA. For any \textsc{seat arrangement} instance that has utility type S-utility, MWA is NP-hard even with symmetric and binary preferences \cite{bodlaender_hedonic_2020}, and, if Spanning Subgraph Isomorphism of regular graphs is NP-hard, so is MUA, even with symmetric and binary preferences \cite{bodlaender_hedonic_2020}. We improve the result for MUA and conclude that the problem is definitely NP-hard. We also show that, with S-utility, both MWA and MUA remain NP-hard even with symmetric and strict preferences. Finally, we extend these results to B-utility and W-utility.

\begin{theorem}
For either B-utility, S-utility, or W-utility, MWA and MUA are NP-hard, even if preferences are binary and symmetric, or symmetric and strict, and the seat graph is planar and has bounded degree.
\label{thm:MaxMaxminSymm}
\end{theorem}
\begin{proof}
Start first with B-utility. We reuse the reduction from Theorem \ref{thm:EfbUtilityTies}. Recall that, in that reduction, from Partition Into Triangles to EFA, there exists an envy-free arrangement $\sigma$ if and only every agent has utility $1$ under $\sigma$, which happens if and only if there exists a partition into triangles in the Partition Into Triangles instance $G$. Thus, in our \textsc{seat arrangement} instance $I = (P,F,H,B)$, with $n$ agents that have binary and symmetric preferences, there exists an arrangement $\sigma$ in which $\sum_{p \in P} u_{\sigma}(p) = n$ if and only if there exists an arrangement $\sigma$ where every agent has utility at least $1$. Furthermore, there exists an arrangement $\sigma$ where every agent has utility at least $1$ if and only if there exists a partition into triangles in the PIT instance $G$. Thus, both MWA and MUA are NP-hard.

If preferences are symmetric and strict, we can modify the valuations from Theorem \ref{thm:EfbUtilityTies} by using arbitrarily small, positive and pairwise distinct $\varepsilon_i$, in a manner similar to that of Theorem \ref{thm:EfsUtility}. With these new preferences, it is still the case that there exists an arrangement $\sigma$ in which $\sum_{p \in P} u_{\sigma}(p) = n$ if and only if there exists an arrangement $\sigma$ where every agent has utility at least $1$, exists if and only if there exists a partition into triangles in $G$. Therefore, under B-utility, both MWA and MUA are NP-hard even if preferences are symmetric and strict.

In the case of W-utility, we use the same reduction. Like in the previous case, if preferences are binary and symmetric there exists an arrangement $\sigma$ in which $\sum_{p \in P} u_{\sigma}(p) = n$ if and only if there exists an arrangement $\sigma$ where every agent has utility at least $1$, which can only happen if there is a partition into triangles in $G$. Thus, an arrangement with total utility $n$ exists if and only if a partition into triangles exists, proving that MWA is NP-hard, while an arrangement with minimum utility $1$ exists if and only if an arrangement into triangles exists, proving that MUA is NP-hard. By using arbitrarily small, positive and pairwise distinct $\varepsilon_i$ to arbitrarily break ties in preferences, we prove that both results hold if preferences are symmetric and strict.

For the case of S-utility, binary and symmetric preferences, and MWA, our goal is to find an arrangement $\sigma$ in which $\sum_{p \in P} u_{\sigma}(p) = 2|E(H)|$, where $E(H)$ is the set of edges of the seat graph $H$. This can happen if and only if the score of every edge is $1$ (recall that edges are counted twice in the previous sum). The proof of Theorem \ref{thm:EfbUtilityTies} shows that a partition into triangles in $G$ implies an arrangement in which the score of every edge is $1$, and that an envy-free arrangement under B-utility implies that the score of every edge is $1$, which in turns implies the existence of a partition into triangles in $G$. Thus, there exists an arrangement $\sigma$ in which $\sum_{p \in P} u_{\sigma}(p) = 2|E(H)|$ if and only if there exists a partition in triangles in $G$, so MWA is NP-hard. The result holds if preferences are symmetric and strict, proven by arbitrarily breaking ties in preferences.

What remains is the case with S-utility, binary and symmetric preferences (or symmetric and strict), and MUA. Our goal is to find an arrangement $\sigma$ where every agent gets utility $1$ or more, which happens if every agent gives a valuation of $1$ to at least one neighbour. By Claim $2$ of Theorem \ref{thm:EfbUtilityTies} this can happen if and only if there is a partition into triangles in $G$. Like in the previous cases, we can show that the result holds for symmetric and strict preferences by arbitrarily breaking ties.
\end{proof}

In fact, in the case of B-utility, both MWA and MUA are NP-hard, even if preferences are $1$-dimensional and strict. To prove it, we will use the reduction from \textsc{Bin Packing} used in Theorem \ref{thm:EnvFree1DBUtil}.

\begin{proposition}
Let $I = (P,F,G,B)$ be a \textsc{seat arrangement} instance with B-utility, $1$-dimensional strict preferences with unique positions, and a seat graph in which every connected component is a path with the same number of vertices. Then both MWA and MUA are NP-hard.
\label{prop:MaxMaxMinBUtil1D}
\end{proposition}
\begin{proof}
\vspace{-5mm}
From an instance of \textsc{Bin Packing}, with item set $X$, item size set $S$, bin capacity $C$, and number of bins $K$ we create a \textsc{seat arrangement} instance $I = (P,F,G,B)$ using the reduction from Theorem \ref{thm:EnvFree1DBUtil}. In general, we have that, if preferences are $1$-dimensional, for any two agents $p,q$, $f_p(q) = f_q(p) = D - d(p,q) + 1$, where $D$ is the maximum distance between agents, and $d(p,q)$ is the distance from $p$ to $q$ in the $1$-dimensional instance. Given how preferences are constructed, for agents $p_i^l$ and $p_i^{l+1}$, which are created from the same element $i \in I$, $f_{p_i^l}(p_i^{l+1}) = f_{p_i^{l+1}}(p_i^l) > D - (1 + n \cdot \varepsilon) + 1 = D - n \cdot \varepsilon$, where $n = |P|$, the number of agents in $I$, and $\varepsilon$ is the arbitrarily small but strictly positive constant used in the proof of Theorem \ref{thm:EnvFree1DBUtil}. For any other agents $p$ and $q$, $f_p(q) = f_q(p) \leq D - 2 + 1 = D - 1$. Thus, there exists a matching in which every agent has utility at least $D - n \cdot \varepsilon$ if and only if agents can be arranged in an envy-free arrangement like the one described in Theorem \ref{thm:EnvFree1DBUtil}, which happens if an only if the elements in $X$ can be packed into bins. Hence, MUA is NP-hard. 

As for MWA, there exists an arrangement $\sigma$ in which the sum of all utilities $\sum_{p \in P} u_{\sigma}(p)$ is at least $n\cdot(D - n \cdot \varepsilon)$ if an only if the utility of every agent is at least $D - n \cdot \varepsilon$, which in turn happens if and only if the elements in $X$ can be packed into bins, proving that MWA is NP-hard. 
\end{proof}

As in the case of Theorem \ref{thm:EnvFree1DBUtil}, this result also applies when the seat graph is a cycle or cluster graph.

\section{Conclusion}
\label{sec:Conc}
We have introduced two new types of utilities to \textsc{seat arrangement}, namely B-utility and W-utility, giving us a more complete picture of the algorithmic behaviour of the problem for different utility functions. In particular, we have seen that, in general, the complexity of a given variant of \textsc{seat arrangement} does not change between S-utility and both B-utility and W-utility. However, the computational complexity of some problems differs between the different utility types. One example is finding an exchange-stable arrangement when preferences are symmetric, because with S-utility it is PLS-complete, but with B-utility it is solvable in polynomial time. Additionally, we have studied a type of agent preference, derived from a one dimensional space, which we call $1$-dimensional preferences. In this case, if the seat graph contains a cycle or path component, and the utility type is S-utility, no envy-free arrangement can exist. Meanwhile, if preferences are $1$-dimensional and the utility type is $B$-utility, EFA, MWA, and MUA are NP-hard.

Some possible future research directions include studying the open variants of \textsc{seat arrangement} from Table \ref{tab:resultSummary}, particularly those where preferences are $1$-dimensional. We hypothesize that, with W-utility and $1$-dimensional preferences, EFA is NP-hard if the seat graph is a path graph, while an envy-free arrangement cannot exist if the seat graph contains a cycle as a connected component. On the other hand, we believe that, with S-utility and W-utility, MWA and MUA are solvable in polynomial time if preferences are $1$-dimensional.

Another possibility concerns the creation of Integer Programming or Constraint Programming models for NP-hard variants of \textsc{seat arrangement}. There has been very little empirical work on \textsc{seat arrangement}, and, given how most problems are computationally hard, it is of interest to study whether integer or constraint programming solvers can efficiently find solutions in moderately-sized instances.

\bibliographystyle{abbrv}
\bibliography{bibliography}


\end{document}